\documentclass[11pt]{article}

\usepackage{fullpage}
\usepackage{amsthm, amssymb}
\usepackage{amsmath}
\usepackage{multirow}
\usepackage{tabularx}
\usepackage{color}
\usepackage[ruled,vlined,linesnumbered]{algorithm2e}

\usepackage{ifpdf}
\ifpdf    
\usepackage{hyperref}
\else    
\usepackage[hypertex]{hyperref}
\fi

\def\compactify{\itemsep=0pt \topsep=0pt \partopsep=0pt \parsep=0pt}

\newtheorem{theorem}{Theorem}[section]
\newtheorem{lemma}[theorem]{Lemma}

\newtheorem{claim}[theorem]{Claim}

\newtheorem{corollary}[theorem]{Corollary}

\newtheorem{definition}[theorem]{Definition}

\newcommand{\ignore}[1]{}

\def\Holder{H\"older}
\newcommand{\zo}{\{0,1\}}

\newcommand{\N}{{\mathbb{N}}}

\newcommand{\R}{{\mathbb{R}}}

\newcommand{\eps}{\epsilon}

\newcommand{\Var}[1]{{\bold{Var}\left[#1\right]}}
\newcommand{\Enop}[1]{{\mathbb{E}{#1}}}
\newcommand{\E}[1]{{\mathbb{E}\left[{#1}\right]}}
\newcommand{\EE}[2]{{\mathbb{E}_{#1}\left[{#2}\right]}}

\DeclareMathOperator{\W}{{\mathcal W}}

\DeclareMathOperator{\median}{median}

\def\eqdef{\stackrel{\text{\tiny\rm def}}{=}}

\newcommand{\aset}[1]{\{ #1 \}}
\DeclareMathOperator{\dx}{dx}
\DeclareMathOperator{\du}{du}

\title{Streaming Algorithms via Precision Sampling}

\author{
Alexandr Andoni\thanks{Work done in part while the author was a
  postdoctoral researcher at Princeton University/CCI, supported by NSF CCF 0832797.}\\
 Microsoft Research
\and
Robert Krauthgamer\thanks{Supported in part by The Israel Science Foundation (grant \#452/08), and by a Minerva grant.}\\
Weizmann Institute
\and
Krzysztof Onak\thanks{Supported in part by
 a Simons Postdoctoral Fellowship and NSF grants 0732334 and 0728645.}\\
CMU
}

\definecolor{darkgreen}{rgb}{0,.4,0}

\newcommand{\rnote}[1]{}
\newcommand{\aanote}[1]{}
\newcommand{\knote}[1]{}

\newcommand{\myparagraph}[1]{\medskip\noindent{\bf #1}}

\begin{document}

\setcounter{page}{0}
 
\maketitle
\thispagestyle{empty}

\begin{abstract}
A technique introduced by Indyk and Woodruff [STOC 2005] has
inspired several recent advances in data-stream algorithms. 
We show that a number of these results follow easily from
the application of a single probabilistic method 
called {\em Precision Sampling}. 
Using this method, we obtain simple data-stream algorithms that maintain
a randomized sketch of an input vector $x=(x_1,\ldots x_n)$, 
which is useful for the following applications:
\begin{itemize}
\item
Estimating the $F_k$-moment of $x$, for $k>2$.
\item
Estimating the $\ell_p$-norm of $x$, for $p\in[1,2]$, with small update time.
\item
Estimating cascaded norms $\ell_p(\ell_q)$ for all $p,q>0$.
\item
$\ell_1$ sampling, where the goal is to produce an element $i$ with
probability (approximately) $|x_i|/\|x\|_1$. It extends to similarly
defined $\ell_p$-sampling, for $p\in [1,2]$. 
\end{itemize}
For all these applications the algorithm is essentially the same: 
scale the vector $x$ entry-wise by a well-chosen random vector, and run a
heavy-hitter estimation algorithm on the resulting vector.
Our sketch is a linear function of $x$,  thereby allowing general
updates to the vector $x$. 

Precision Sampling itself addresses the problem of estimating
a sum $\sum_{i=1}^n a_i$ from weak estimates of each real $a_i\in[0,1]$.
More precisely, the estimator first chooses a desired precision
$u_i\in(0,1]$ for each $i\in[n]$,
and then it receives an estimate of every $a_i$ within additive $u_i$. 
Its goal is to provide a good approximation to $\sum a_i$
while keeping a tab on the ``approximation cost'' $\sum_i (1/u_i)$.
Here we refine previous work [Andoni, Krauthgamer, and Onak, FOCS 2010]
which shows that as long as $\sum a_i=\Omega(1)$, a good multiplicative 
approximation can be achieved using total precision of only $O(n\log n)$.
\end{abstract}

\newpage

\section{Introduction}

A number of recent developments in algorithms for data streams 
have been inspired, at least in part, by a technique devised 
by Indyk and Woodruff \cite{IW05} to obtain
near-optimal space bounds for estimating $F_k$ moments, for
$k>2$. Indeed, refinements and modifications of that technique were used for
designing better or new algorithms for applications such as:
$F_k$ moments~\cite{Ganguly-Fk06} (with better bounds than~\cite{IW05}), entropy
estimation~\cite{BG-entropy06}, cascaded
norms~\cite{GBD-hybrid08, JW-cascaded},
Earthmover Distance~\cite{ADIW-EMD}, $\ell_1$ sampling algorithm
\cite{MW-l1sampling}, distance to independence of two random
variables~\cite{BO-independence10}, and even, more generically, a
characterization of ``sketchable''  functions of
frequencies~\cite{BO-zeroonelaw10}. While clearly very powerful, the
Indyk--Woodruff technique is somewhat technically involved, and
hence tends to be cumbersome to work with.

In this paper, we show an alternative design for the Indyk--Woodruff
technique, resulting in a simplified algorithm for several of the above
applications. Our key ingredient, dubbed the {\em Precision
Sampling Lemma (PSL)}, is a probabilistic method, concerned with 
estimating the sum of a number of real quantities. The PSL
was introduced in~\cite[Lemma 3.12]{AKO-edit}, in an unrelated context, of
{\em query-efficient} algorithms (in the sense of property testing) 
for estimating the edit distance.

Our overall contribution here is providing a generic approach that
leads to simplification and unification of a family of
data-stream algorithms. Along the way we obtain new and improved bounds
for some applications. We also give a slightly improved version of
the PSL.

In fact, all our algorithms comprise of the following two simple
steps: multiply the stream by well-chosen random numbers (given by
PSL), and then solve a certain heavy-hitters problem. Interestingly,
each of the two steps (separately) either has connections to or is a
well-studied problem in the literature of data streams.  Namely, our
implementation of the first step is somewhat similar to {\em Priority
  Sampling} \cite{DLT07-prioritySampling}, as discussed in
Section~\ref{sec:prioritySampling}. The second step, the heavy-hitters
problem, is a natural streaming primitive, studied at least since the
work of Misra and Gries~\cite{MG82}. It would be hard to list all the
relevant literature for this problem concisely; instead we refer the
reader, for example, to the survey by Muthukrishnan~\cite{Muthu-book}
and the CountMin wiki site~\cite{countminsketch} and the references
therein.

\subsection{Streaming Applications}

We now describe the relevant streaming applications in detail. In most cases, the
input is a vector $x\in \R^n$, which we maintain under stream
updates. An update has the form $(i,\delta)$, which means that $\delta \in \R$ 
is added to $x_i$, the $i$th coordinate of $x$.\footnote{We make a standard
discretization assumption that all numbers have a finite precision, and
in particular, $\delta\in \{-M,-M+1,\ldots, M-1,M\}$,
for $M=n^{O(1)}$.} The
goal is to maintain a sketch of $x$ of small size (much smaller than
$n$), such that, at the end of the stream, the algorithm outputs some
function of $x$, depending on the actual problem in mind. Besides the
space usage, another important complexity measure is the update
time --- how much time it takes to modify the sketch to reflect an update $(i,\delta)$.

We study the following problems.\footnote{Since we work in the general
  update framework, we will not be presenting
  the literature that is concerned with restricted types of updates,
  such as positive updates $\delta>0$.} 
For all these problems, the algorithm is essentially the same (see the
beginning of Section~\ref{sec:basicStreamingApps}). 
All space bounds are in terms of words, each having $O(\log n)$ bits.

\begin{description}
\item[$\mathbf{F_k}$ moment estimation, for $\mathbf{k>2}$:] 
The goal is to produce a $(1+\eps)$ factor approximation to the $k$-th moment
of $x$, i.e., $\|x\|_k^k=\sum_{i=1}^n |x_i|^k$. 
The first sublinear-space algorithm for $k > 2$, due to \cite{AMS}, 
gave a space bound $n^{1-1/k}\cdot (\eps^{-1}\log n)^{O(1)}$,
and further showed the first polynomial lower bound for $k$ sufficiently large.
A lower bound of $\Omega(n^{1-2/k})$ was shown in \cite{ChakrabartiKS03,BJKS},
and it was (nearly) matched by Indyk and Woodruff~\cite{IW05}, who
gave an algorithm using space $n^{1-2/k}\cdot (\eps^{-1}\log n)^{O(1)}$. Further
research reduced the space bound to essentially 
$O(n^{1-2/k}\cdot \eps^{-2-4/k}\log^2 n)$ ~\cite{Ganguly-Fk06,MW-l1sampling}
(see \cite{MW-l1sampling} for multi-pass bounds).
Independently of our work, this bound was improved by a roughly
$O(\log n)$  factor in \cite{BO10}.

Our algorithm for this problem appears in Section~\ref{sec:FkMoments},
and improves the space usage over these bounds. Very recently,
following the framework introduced here, \cite{Ganguly-personal}
reports a further improvement in space for a certain regime of
parameters.

\item[$\mathbf{\ell_p}$-norm estimation, for $\mathbf{p\in[1,2]}$:]
The goal is to produce a $1+\eps$ factor approximation to $\|x\|_p$,
just like in the previous problem.%
\footnote{The difference in notation ($p$ vs.\ $k$) is partly due to
  historical reasons: the $\ell_p$ norm for
  $p\in[1,2]$ has been usually studied separately from the $F_k$ moment
  for $k>2$, having generally involved somewhat different
  techniques and space bounds. }
The case $p=2$, i.e., $\ell_2$-norm
estimation was solved in \cite{AMS}, which gives a space bound of
$O(\eps^{-2}\log n)$. 
It was later shown in \cite{I00b} how to estimate $\ell_p$ norm for all
$p\in(0,2]$, using $p$-stable distributions, in $O(\eps^{-2}\log n)$
space. Further research aimed to get a tight bound and
to reduce the update time (for small $\eps$) 
from $\Omega(\eps^{-2})$ to $\log^{O(1)}n$
(or even $O(1)$ for $p=2$), 
see, e.g.,~\cite{NW-fastsketches10,
  KNW-exact10,Li-estimators08, GC-moments07} and references therein. 

Our algorithm for this problem appears in Section~\ref{sec:ell1-norm}
for $p=1$ and Section~\ref{sec:ellp-norm} for all $p\in[1,2]$. The
algorithm has an improved update time, over that of~\cite{GC-moments07}, for
$p\in(1,2]$, and uses comparable space, $O(\eps^{-2-p}\log^2 n)$.
We note that, for $p=1$, our space bound is worse
than that of~\cite{NW-fastsketches10}. 
Independently of our work, fast space-optimal algorithms for all $p\in(0,2)$
were recently obtained in \cite{KNPW10-otimalUpdate}.

\item[Mixed/cascaded norms:]
The input is a matrix $x\in \R^{n\times n}$, and the goal
is to estimate the $\ell_p(\ell_q)$ norm, defined as 
$\|x\|_{p,q}=\left(\sum_{i\in[n]} (\sum_{j\in[n]} |x_{i,j}|^q)^{p/q}\right)^{1/p}$,
for $p,q\ge 0$. Introduced
in~\cite{CM05}, this problem
generalizes the $\ell_p$-norm/$F_k$-moment estimation questions, and 
for various values of $p$ and $q$, it has particular
useful interpretations, see~\cite{CM05} for examples. Perhaps the
first algorithm, applicable to some regime of
parameters, appeared in \cite{GBD-hybrid08}. Further progress on the
problem was accomplished 
in~\cite{JW-cascaded}, which obtains near-optimal bounds
for a large range of values of $p,q\ge 0$
(see also~\cite{MW-l1sampling} and~\cite{GBD-hybrid08}).

We give in Section~\ref{sec:mixednorms} algorithms for all parameters
$p,q>0$, and obtain bounds that are tight up to $(\eps^{-1}\log
n)^{O(1)}$ factors.  In particular, we obtain the first algorithm for
the regime $q>p>2$ --- no such (efficient) algorithm was previously
known. We show that the space complexity is controlled by a metric
property, which is a generalization of the {\em $p$-type constant} of
$\ell_q$. Our space bounds fall out directly from bounds on this property.

\item[$\mathbf{\ell_p}$-sampling, for $\mathbf{p\in[1,2]}$:]
Here, the goal of the algorithm is to produce an index $i\in[n]$
sampled from a distribution $D_x$ that depends on $x$, as opposed to
producing a fixed function of $x$. In particular, the (idealized) goal is to
produce an index $i\in[n]$ where each $i$ is returned with probability
$|x_i|^p/\|x\|_p^p$. We meet this goal in an approximate fashion: 
there exists some approximating distribution $D'_x$ on $[n]$, where
$D'_x(i)=(1\pm \eps)|x_i|/\|x\|_1\pm 1/n^2$ (the exponent 2 here is
arbitrary), such that the algorithm outputs $i$ drawn from the
distribution $D'_x$. 
Note that the problem would
be simple if the stream had only insertions (i.e., $\delta\ge0$
always); so the challenge is to be able to support both positive and 
negative updates to the vector $x$.

The $\ell_p$-sampling problem was introduced in \cite{MW-l1sampling},
where it is shown that the $\ell_p$-sampling problem is a useful
building block for other streaming problems, including cascaded norms,
heavy hitters, and moment estimation. The algorithm in
\cite{MW-l1sampling} uses $(\eps^{-1}\log n)^{O(1)}$ space.

Our algorithm for the $\ell_p$-sampling problem, for $p\in[1,2]$,
appears in Section~\ref{sec:ellp-sampling}.  It improves the space to
$O(\eps^{-p}\log^3 n)$.  Very recently, following the framework
introduced here, \cite{JST10} further improve the space bound to a {\em
  near-optimal} bound, and extend the algorithm to $p\in[0,1]$.
\end{description}

All our algorithms maintain a linear sketch $L:\R^n\to \R^S$
(i.e. $L$ is a linear function), where
$S$ is the space bound (in words, or $O(S\log n)$ in bits). Hence, all
the updates may be implemented 
using the linearity: $L(x+\delta e_i)=Lx+\delta\cdot Le_i$,
where $e_i$ is the $i$th standard basis vector.

\subsection{Precision Sampling} \label{sec:PSL}

We now describe the key primitive used in all our algorithms, the
Precision Sampling Lemma (PSL). It has originally
appeared in~\cite{AKO-edit}. The present version is improved in two
respects: it has better bounds and is streaming--friendly.

PSL addresses a variant of the standard sum--estimation problem, 
where the goal is to estimate the sum $\sigma\eqdef \sum_i a_i$ of $n$ 
unknown quantities $a_i\in[0,1]$.
In the standard sampling approach, one randomly samples a set of indices
$I\subset [n]$, 
and uses these $a_i$'s to compute an estimate 
such as $\tfrac{n}{|I|} \sum_{i\in I}a_i$. 
{\em Precision sampling} considers a different scenario, 
where 
the estimation algorithm chooses a sequence of \emph{precisions}
$u_i\in(0,1]$ (without knowing the $a_i$'s), and then obtains a
  sequence of estimates $\hat a_i$ that satisfy $|\hat a_i-a_i|\le
  u_i$, and it has to report an estimate for the sum $\sigma = \sum_i
  a_i$.  As it turns out from applications, producing an estimate with
  additive error $u_i$ (for a single $a_i$) incurs cost $1/u_i$, hence
  the goal is to achieve a good approximation to $\sigma$ while
  keeping tabs on the total cost (total precision) $\sum_i
  (1/u_i)$.\footnote{Naturally, in other application, other notions of
    cost may make more sense, and are worth investigating.}

To illustrate the concept, consider the case where $10\le\sigma\le 20$, and
one desires a $1.1$ multiplicative approximation to $\sigma$. 
How should one choose the precisions $u_i$? One approach is
to employ the aforementioned sampling approach: choose a random set of indices 
$I\subset [n]$ and assign to them a high precision, say $u_i=1/n$, 
and assign trivial precision $u_i=1$ to the rest of indices;
then report the estimate $\hat\sigma = \tfrac{n}{|I|}\sum_{i\in I} \hat a_i$.
This way, the error due to the adversary's response is at most 
$\tfrac{n}{|I|}\sum_{i\in I}|\hat a_i-a_i| \le 1$,
and standard sampling (concentration) bounds prescribe setting $|I|=\Theta(n)$.
The total precision becomes $\Theta(n\cdot |I|)=\Theta(n^2)$,
which is no better than naively setting all precisions 
$u_i=1/n$, which achieves total additive error $1$ using total precision $n^2$.
Note that in the restricted case where all $a_i\leq 40/n$, 
the sampling approach is better, because setting $|I|=O(1)$ suffices; 
however, in another restricted case where all $a_i\in\zo$,
the naive approach could fare better, if we set all $u_i=1/2$.
Thus, total precision $O(n)$ is possible in both cases, 
but by a different method.
We previously proved in~\cite{AKO-edit} that one can always choose
$w_i$ randomly such that $\sum w_i\leq O(n\log n)$ with constant
probability.

In this paper, we provide a more efficient version of PSL
(see Section \ref{sec:PSLproof} for details).
To state the lemma, we need a definition
that accommodates both additive and multiplicative errors.

\begin{definition} [Approximator]
\label{def:approximator}
Let $\rho>0$ and $f\in[1,2]$.
A {\em $(\rho,f)$-approximator} to $\tau>0$ is any quantity $\hat \tau$ 
satisfying
$
\tau/f-\rho\le \hat \tau\le f\tau+\rho.
$
(Without loss of generality, $\hat\tau\ge 0$.)
\end{definition}

The following lemma is stated in a rather general form.
Due to historical reasons, the lemma refers to precisions as $w_i\in[1,\infty)$,
which is identical to our description above via $w_i=1/u_i$.
Upon first reading, it may be instructive to consider the special case
$f=1$, and let $\rho=\eps>0$ be an absolute constant
(say $0.1$ to match our discussion above).

\begin{lemma}[Precision Sampling Lemma]
\label{lem:nonUniformSampling}
Fix an integer $n\ge 2$, a multiplicative error $\eps\in[1/n,1/3]$, and
an additive error $\rho\in[1/n,1]$. Then there exist a
distribution $\W$ on the real interval $[1,\infty)$
and a 
reconstruction algorithm $R$, 
with the following two properties.
\begin{description} \compactify
\item [Accuracy:]
Consider arbitrary $a_1,\ldots,a_n\in[0,1]$ and $f\in[1,1.5]$. 
Let $w_1,\ldots,w_n$ be chosen at random from $\W$
pairwise independently.%
\footnote{That is, for all $i<j$, the pair $(w_i,w_j)$ is distributed 
as $\W^2$.
}
Then with probability at least $2/3$,
when algorithm $R$ is given $\{w_i\}_{i\in[n]}$ and $\{\hat a_i\}_{i\in[n]}$ 
such that each $\hat a_i$ is an arbitrary $(1/w_i,f)$-approximator of $a_i$,
it produces $\hat \sigma\ge 0$ which is
a $(\rho,f\cdot e^{\eps})$-approximator to $\sigma \eqdef \sum_{i=1}^n a_i$.
\item [Cost:] There is $k=O(1/\rho\eps^2)$ such that 
the conditional expectation
$\EE{w\in \W}{w \mid M}\le O(k\log n)$ for some event $M=M(w)$
occurring with high probability. For every fixed $\alpha\in (0,1)$, we have $\EE{w\in
  \W}{w^\alpha}\le O(k^\alpha)$. 
The distribution
  $\W=\W(k)$ depends only on $k$.
\end{description}
\end{lemma}

We emphasize that the probability $2/3$ above 
is over the choice of $\{w_i\}_{i\in[n]}$
and holds (separately) for every fixed setting of $\{a_i\}_{i\in[n]}$.
In the case where $R$ is randomized, the probability $2/3$
is also over the coins of $R$.
Note also that the precisions $w_i$ are chosen without knowing $a_i$,
but the estimators $\hat a_i$ are adversarial ---
each might depend on the entire $\{a_i\}_{i\in[n]}$ and $\{w_i\}_{i\in[n]}$,
and their errors might be correlated.

In our implementation, it turns out that the reconstruction algorithm
uses only $\hat{a}_i$'s which are (retrospectively) good approximation
to $a_i$ --- namely $\hat{a}_i\gg 1/w_i$ --- hence the adversarial effect is limited. For completeness, we also mention that, for $k=1$, the
distribution $\W=\W(1)$ is simply $1/u$ for a random $u\in[0,1]$. We
present the complete proof of the lemma in Section~\ref{sec:PSLproof}.

It is natural to ask whether the above lemma is tight. In
Section~\ref{sec:lowerBound}, we show a lower bound on $\EE{w\in
  \W}{w}$ in the considered setting, which matches our PSL bound up to
a factor of $1/\eps$. We leave it as an open question what is the best
achievable bound for PSL.

\subsection{Connection to {\em Priority} Sampling}
\label{sec:prioritySampling}

We remark that (our implementation of) Precision Sampling has some
similarity to {\em Priority Sampling} \cite{DLT07-prioritySampling}, which
is a scheme for the following problem.%
\footnote{The similarity is at the more technical level of applying
  the PSL in streaming algorithms, hence the foregoing discussion actually 
  refers to Sections~\ref{sec:PSLproof} and~\ref{sec:basicStreamingApps}.
}
We are given a
vector $x\in \R_+^n$ of positive weights (coordinates), and we want to maintain a sample of $k$
weights in order to be able to estimate sums of weights for an
arbitrary subset of coordinates, i.e., $\sum_{i\in I} x_i$ for arbitrary
sets $I\subseteq [n]$. Priority Sampling has been shown to
attain an essentially best possible variance for a sampling
scheme \cite{Sze06-DLPoptimality}. 

The similarity between the two sampling schemes is the following. In our main
approach, similarly to the approach in Priority Sampling, we take the vector
$x\in \R^n$, and consider a vector $y$ where $y_i=x_i/u_i$, for $u_i$
chosen at random from $[0,1]$. We are then interested
in heavy hitters of the vector $y$ (in $\ell_1$ norm). We obtain these
using the CountSketch/CountMin
sketch~\cite{CCF, CM05b}. In Priority Sampling,
one similarly extracts a set of $k$ heaviest coordinates of
$y$. However, one important
difference is that in Priority Sampling the weights (and updates) are positive, thus
making it possible to use Reservoir sampling-type techniques to obtain
the desired heavy hitters. In contrast, in our setting the weights
(and updates) may be negative,
and we need to extract the heavy hitters
approximately and hence post-process them differently. 

See also~\cite{CDKLT09-streamSampling} and the references therein for
streaming-friendly versions of Priority Sampling and other related sampling
procedures.\aanote{need to take a CLOSER LOOK!}

\section{Proof of the Precision Sampling Lemma}
\label{sec:PSLproof}

In this section we prove the Precision Sampling Lemma
(Lemma~\ref{lem:nonUniformSampling}). Compared to our previous version
of PSL from~\cite{AKO-edit}, this version has the following improvements: 
a better bound on $\EE{w\in \W}{w}$ (hence better total precision), 
it requires the $w_i$'s to be only pairwise independent (hence streaming-friendly), 
and a slightly simpler construction and analysis via its inverse $u=1/w$.  
We also show a lower bound in Section~\ref{sec:lowerBound}.

\myparagraph{The probability distribution $\W$.}
Fix $k={\zeta}/{\rho \eps^2}$ for sufficiently large constant $\zeta>0$.
The distribution $\W$ takes a random value $w\in [1,\infty)$ as
  follows: pick i.i.d.\ samples $u_1,\ldots,u_k$ from the uniform
  distribution $U(0,1)$, and set $w \eqdef 
\max_{j\in[k]} 1/u_j$.  Note that $\W$ depends on $k$ only.

\myparagraph{The reconstruction algorithms.}
The randomized reconstruction algorithm $R'$
gets as input $\{w_i\}_{i\in[n]}$ and $\{\hat a_i\}_{i\in[n]}$
and works as follows. For each $i\in[n]$, 
sample $k$ i.i.d.\ random variables, $u_{i,j}\in U(0,1)$ for $j\in[k]$,
conditioned 
on the event $\aset{w_i=\max_{j\in[k]}1/u_{i,j}}$.
Now define the ``indicators'' $s_{i,j}\in\{0,1/k\}$, for each $i\in[n],j\in[k]$, 
by setting 
\begin{equation*}
  s_{i,j}\eqdef\begin{cases}
    1/k & \text{if $u_{i,j}\le \hat a_i/t$ for $t\eqdef 4/\eps$}; 
    \\ 
    0 & \text{otherwise}.
  \end{cases}
\end{equation*}
Finally, algorithm $R'$ sets $s \eqdef \sum_{i\in [n], j\in[k]} s_{i,j}$ 
and reports $\hat\sigma \eqdef s\, t$
as an estimate for $\sigma=\sum_i a_i$. 
A key observation is that altogether, 
i.e., when we consider both the coins involved in the choice of $w_i$ from $\W$
as well as those used by algorithm $R'$,
we can think of $u_{i,1},\ldots,u_{i,k}$ as being chosen 
i.i.d.\ from $U(0,1)$.
Observe also that whenever $\hat a_i$ is a $(1/w_i,f)$-approximator to $a_i$,
it is also a $(u_{i,j}\,,f)$-approximator to $a_i$ for \emph{all} $j\in[k]$.

We now build a more efficient deterministic algorithm $R$ that
performs at least as well as $R'$. Specifically, $R$ does not generate
the $u_{i,j}$'s (from the given $w_i$'s), but rather sets $s_i \eqdef
\E{\sum_{j\in[k]}s_{i,j}\mid \min_{j\in [k]} u_{i,j}=1/w_i}$ and $s
\eqdef \sum_{i\in[n]} s_i$.  A simple calculation yields an explicit
formula, which is easy to compute algorithmically:
\begin{equation*}
  s_i = \begin{cases}
    \tfrac1k+\tfrac{k-1}{k}\cdot\tfrac{\hat a_iw_i/t-1}{w_i-1}; 
    & \text{if $\hat a_iw_i/t\ge 1$}\\
    0 & \text{otherwise}.
  \end{cases}
\end{equation*}
We proceed to the analysis of this construction.
We will first consider the randomized algorithm $R'$,
and then show that derandomization can only decrease the error.

\begin{proof}[Proof of Lemma~\ref{lem:nonUniformSampling}]
We first give bounds on the moments of the distribution $\W$.  Indeed,
recall that by definition $w \eqdef \max_{j\in[k]}
\tfrac{1}{u_{j}}$. We define the event $M$ to be that $w\le n^{5}$;
note that $\Pr[M]\ge 1-k\cdot n^{-5}\ge 1-O(n^{-2})$.  Conditioned on
$M$, each $u_j\in U(n^{-5},1)$, and we have $ \E{\tfrac{1}{u_j}} =
\tfrac{1}{1-n^{-5}} \int_{n^{-5}}^1 \tfrac{1}{x}\dx =
\tfrac{\ln(n^5)}{1-n^{-5}}$. Thus
\[
  \EE{w\in\W}{w\mid M} 
  \leq \E{\textstyle \sum_{j\in[k]} \tfrac{1}{u_j} \mid M}
  \leq O(k\log n).
\]
Now fix $\alpha\in(0,1)$. It is immediate that
$\E{1/u^\alpha}=O(1/(1-\alpha))$. We can similarly prove that $
\EE{w\in \W}{w^\alpha} \leq O(k^\alpha/(1-\alpha))$, but the calculation
is technical, and we include its proof in
Appendix~\ref{apx:expAlpha}.

We now need to prove that $\hat\sigma$ is an approximator to $\sigma$,
with probability at least $2/3$. The plan is to first compute the
expectation of $s_{i,j}$, for each $i\in [n], j\in[k]$. This
expectation depends on the approximator values $\hat a_i$, which
itself may depend (adversarially) on $w_i$, so instead we give upper
and lower bounds on the expectation $\E{s_{i,j}}\approx \tfrac{a_i}{tk}$.
Then, we wish to apply a concentration bound on the sum of $s_{i,j}$,
but again the $s_{i,j}$ might depend on the random values $w_i$, so we
actually apply the concentration bound on the upper/lower bounds of
$s_{i,j}$, and thereby derive bounds on $s=\sum s_{i,j}$.

Formally, we define random variables $\overline{s}_{i,j},
\underline{s}_{i,j}\in\{0,1/k\}$. We set $\overline{s}_{i,j}=1/k$ iff 
$u_{i,j}\le fa_i/(t-1)$, and $0$ otherwise. Similarly, we set
$\underline{s}_{i,j}=1/k$ iff $u_{i,j}\le a_i/f(t+1)$, and $0$ otherwise.
We now claim that 
\begin{equation}
\label{eqn:sUBLB}
\underline{s}_{i,j}\le s_{i,j}\le \overline{s}_{i,j}.
\end{equation}
Indeed, if $s_{i,j}=1/k$ then $u_{i,j} \le \hat a_{i}/t$, and hence,
using the fact that $\hat a_i$ is a $(u_{i,j}\,,f)$-approximator to $a_i$, 
we have $u_{i,j}\le fa_i/(t-1)$, or
$\overline{s}_{i,j}=1/k$. 
Similarly, if $s_{i,j}=0$, then $u_{i,j} > \hat a_{i}/t$, 
and hence $u_{i,j}>a_i/f(t+1)$, or
$\underline{s}_{i,j}=0$. Notice for later use that each of 
$\{\overline{s}_{i,j}\}$ and $\{\underline{s}_{i,j}\}$ 
is a collection of $nk$ pairwise independent random variables.
For ease of notation, define
$\underline{\hat \sigma}={t}\sum_{i,j} \underline{s}_{i,j}$
and 
$\overline{\hat\sigma}={t}\sum_{i,j} \overline{s}_{i,j}$,
and observe that 
$\underline{\hat \sigma}\le \hat\sigma\le \overline{\hat\sigma}$.

We now bound $\E{\overline{s}_{i,j}}$ and
$\E{\underline{s}_{i,j}}$. For this, it suffices to compute the
probability that $\overline{s}_{i,j}$ and $\underline{s}_{i,j}$ are
$1/k$. For the first quantity, we have:
\begin{equation}
\label{eqn:sij-upperBound}
  \Pr\Big[\overline{s}_{i,j}=\tfrac1k\Big]
  = \Pr\Big[ u_{i,j}\le \tfrac{fa_i}{t-1} \Big]
  = \tfrac{fa_i}{t-1} 
  \leq e^{\eps/2}f \cdot \tfrac{a_i}{t},
\end{equation}
where we used the fact that $t-1\ge e^{-\eps/2}t$. 
Similarly, for the second quantity, we have:
\begin{equation}
\label{eqn:sij-lowerBound}
  \Pr\Big[\underline{s}_{i,j}=\tfrac1k \Big]
  = \Pr\Big[ u_{i,j}\le \tfrac{a_i}{f(t+1)} \Big]
  = \tfrac{a_i}{f(t+1)} 
  \geq e^{-\eps/2}f^{-1} \cdot \tfrac{a_i}{t}.
\end{equation}

Finally, using Eqn.~\eqref{eqn:sUBLB} and the fact that
$\E{s}=\sum_{i,j} \E{s_{i,j}}$, we can bound the expectation and
variance of $\hat \sigma=st$ as follows:
\begin{equation}
e^{-\eps/2}f^{-1}\cdot \sigma
\le
{t}\sum_{i,j} \E{\underline{s}_{i,j}}
\le
\E{\hat \sigma}
\le
{t}\sum_{i,j} \E{\overline{s}_{i,j}}
\le
 e^{\eps/2}f\cdot \sigma,
\end{equation}
and, using pairwise independence,
$
\Var{\underline{\hat \sigma}}
,
\Var{\overline{\hat \sigma}}
\le
t^2 \cdot \sum_{i,j}k^{-2}\cdot e^{\eps/2} \cdot \tfrac{fa_i}{t}
\le 4tk^{-1}\sigma.
$
Recall that we want to bound the probability that 
$\underline{\hat \sigma}$ and $\overline{\hat \sigma}$ deviate (additively)
from their expectation by roughly $\eps\sigma+\rho$,
which is larger than their standard deviation 
$O(\sqrt{tk^{-1}\sigma}) = O(\sqrt{\rho\eps\sigma})$.

Formally, to bound the quantity $\hat\sigma$ itself, we distinguish two
cases. First, consider $\sigma>\rho/\eps$. 
Then for our parameters $k={\zeta}/{\rho\eps^2}$ and $t=4/\eps$,
$$
\Pr\Big[\overline{\hat \sigma}>e^{\eps/2}f\sigma\cdot (1+\eps/2)\Big]
\le
\Pr\Big[\overline{\hat \sigma}-\E{\overline{\hat\sigma}} 
     > \eps/2\cdot e^{\eps}f\sigma\Big]
\le
\tfrac{\Var{\overline{\hat \sigma}}}{(\eps/2\cdot e^{\eps}f\sigma)^2}
\le 
\tfrac{4tk^{-1}\sigma}{\eps^2\sigma^2/4}
\le \tfrac{O(\rho/\eps\zeta)}{\sigma}
\le 0.1
$$
for sufficiently large $\zeta$. Similarly, $\Pr[\underline{\hat
  \sigma}<f^{-1}e^{-\eps/2}\sigma\cdot e^{-\eps/2}]\le 0.1$.

Now consider the second case, when $\sigma\le \rho/\eps$. Then we have
$$
\Pr\Big[\overline{\hat \sigma}>fe^{\eps/2}\sigma+\rho\Big]
\le
\Pr\Big[\overline{\hat \sigma}-\E{\overline{\hat\sigma}} > \rho \Big]
\le
\tfrac{\Var{\overline{\hat \sigma}}}{\rho^2}
\le 
\tfrac{4tk^{-1}\cdot\rho/\eps}{\rho^2}
\le 0.1.
$$ Similarly, we have $\Pr[\underline{\hat
    \sigma}<f^{-1}e^{-\eps/2}\sigma-\rho]\le 0.1$.  This completes the
proof that $\hat\sigma$ is a $(\rho,fe^\eps)$-approximator to
$\sigma$, with probability at least $2/3$.

Finally, we argue that switching to the deterministic algorithm $R$
only decreases the variances without affecting the expectations, 
and hence the same concentration bounds hold.
Formally, denote our replacement for $s_i$ by
$s'_i=\EE{u_{i,j}}{\sum_{j\in[k]} s_{i,j}\mid \max_{j\in [k]} 1/u_{i,j}=w_i}$,
and note it is a random variable (because of $w_i$).
Define 
$\overline{s}'_i 
  = \E{\sum_{j\in[k]} \overline{s}_{i,j}\mid \max_{j\in [k]} 1/u_{i,j}=w_i}$,
and by applying conditional expectation to Eqn.~\eqref{eqn:sUBLB}, 
we have $s_i \le \overline{s}'_i$.
We now wish to bound the variance of $\sum_i \overline{s}'_i$.
By the law of total variance, and using the shorthand $\vec{w}=\{w_i\}_i$, 
\begin{equation}
 \label{eqn:varS}
\textstyle
  \Var{\sum_i \overline{s}_i }
  = \E{\Var{\sum_i \overline{s}_i \mid \vec{w}}}
  + \Var{\E{\sum_i \overline{s}_i \mid \vec{w}}}.
\end{equation}
We now do a similar calculation for ${\sum_i \overline{s}'_i }$,
but since each $\overline{s}'_i$ is completely determined from the 
known $\vec{w}$, 
the first summand is just $0$
and in the second summand we can change each $\overline{s}'_i$ to $\overline{s}_i$,
formally
\begin{equation}
\label{eqn:varSprime}
 \textstyle
  \Var{\sum_i \overline{s}'_i }
  = \E{\Var{\sum_i \overline{s}'_i \mid \vec{w}}}
  + \Var{\E{\sum_i \overline{s}'_i \mid \vec{w}}}
  = \Var{\E{\sum_i \overline{s}_i \mid \vec{w}}}.
\end{equation}
Eqns.~\eqref{eqn:varS} and~\eqref{eqn:varSprime} imply that in the
deterministic algorithm the variance (of the upper bound) can indeed
only decrease.  The analysis for the lower bound is analogous, using
$\underline{s}'_i$.  As before, using the fact that the
$\overline{s}_i'$ are pairwise independent (because the $w_i$ are) we
apply Chebyshev's inequality to bound deviation for the algorithm
$R's$ actual estimate $\hat\sigma={t}\sum_i s'_i$.
\end{proof}

\section{Applications I: Warm-Up}
\label{sec:basicStreamingApps}

We now describe our streaming algorithms that use the Precision
Sampling Lemma (PSL) as the core primitive. We first outline two generic
procedures that are used by several of our
applications. The current description leaves some parameters
unspecified: they will be fixed by the particular applications. 
These two procedures are also given in pseudo-code as
Alg.~\ref{alg:ellp-sketch} and Alg.~\ref{alg:ellp-reconstruction}.

As previously mentioned, our sketch function is a linear function
$L:\R^n\to \R^S$ mapping an input vector $x\in\R^n$ into $\R^S$, where
$S$ is the space (in words). The algorithm is simply a fusion of PSL
with a heavy hitters algorithm~\cite{CCF,CM05b}. We use a parameter
$p\ge 1$, which one should think of as the $p$ in the $\ell_p$-norm
estimation problem, and $p=k$ in the $F_k$ moment estimation.  Other
parameters are: $\rho\in (0,1)$ (additive error), $\eps\in(0,1/3)$
(multiplicative error), and $m\in\N$ (a factor in the space usage).

The sketching algorithm is as follows. We start by initializing a vector
of $w_i$'s using Lemma~\ref{lem:nonUniformSampling}: specifically we
draw $w_i$'s from $\W=\W(k)$ for $k=\tfrac{\zeta}{\rho\eps^2}$.
We use $l=O(\log n)$ hash tables $\{H_j\}_{j\in[l]}$, each of size $m$. 
For each hash table $H_j$, 
choose a random hash function $h_j:[n]\to[m]$,
and Rademacher random variables $g_j:[n]\to \{-1,+1\}$.
Then the sketch $Lx$ is obtained by repeating the following for 
every hash table $j\in[l]$ and index $i\in[n]$: 
hash index $i\in[n]$ to find its cell $h_j(i)$, 
and add to this cell's contents the quantity $g_j(i) \cdot x_iw_i^{1/p}$.
Overall, $S=lm$.

The estimation algorithm $E$ proceeds as follows.
First normalize the sketch $Lx$ by scaling it down 
by an input parameter $r\in \R_+$.
Now for each $i\in[n]$, compute the median, over the $l$ hash tables, 
of the $p$th power of cells where $i$ falls into. 
Namely, let $\hat x_i$ be the median of
$|H_j(h_j(i))|^p/r w_i$ over all $j\in[l]$. 
Then run the PSL reconstruction algorithm $R$ on the vectors 
$\{\hat x_i\}_{i\in[n]}$ and $\{w_i\}_{i\in[n]}$, 
to obtain an estimate $\hat\sigma=\hat\sigma(r)$. The final output is $r\cdot
\hat\sigma(r)$.

We note that it will always suffice to use pairwise independence for
each set of random variables $\{w_i\}_i$, $\{g_j(i)\}_i$, and
$\{h_j(i)\}_i$ for each $j\in [l]$.
For instance, it suffices to draw each hash function $h_j$ 
from a universal hash family.

\rnote{do we ever discuss it? we discuss there precision for $\W$, not $Lx$}
\aanote{i guess no, we probably should...}

Finally, we remark that, while the reconstruction
Alg.~\ref{alg:ellp-reconstruction} takes time $\Omega(n)$, one can
reduce this to time $m\cdot (\eps^{-1}\log n)^{O(1)}$ by using a more
refined heavy hitter sketch. We discuss this issue later in this
section.

\begin{algorithm}[!h]
\caption{Sketching algorithm for norm estimation. Input is a vector
  $x\in \R^n$. Parameters $p$, $\eps$, $\rho$, and $m$ are specified later.}
\label{alg:ellp-sketch}
Generate $\{w_i\}_{i\in[n]}$ as prescribed by PSL,
using $\W=\W(k)$ for $k=\zeta\rho^{-1}\eps^{-2}$.
\\
Initialize $l=O(\log n)$ hash tables $H_1,\ldots,H_l$, each of size $m$.
For each table $H_j$, choose a random hash function $h_j:[n]\to [m]$
and a random $g_j:[n]\to \{-1,+1\}$.
\\
\For{each $j\in[l]$}{
Multiply $x$ coordinate-wise with the vectors $\{w_i^{1/p}\}_{i\in[n]}$ and $g_j$, 
and hash the resulting vector into the hash table $H_j$. Formally,
$H_j(z)\triangleq\sum_{i:h_j(i)=z} g_j(i)\cdot w_i^{1/p}\cdot x_i$.
}
\end{algorithm}

\begin{algorithm}[!h]
\caption{Reconstruction algorithm for norm estimation. Input consists
  of $l$ hash tables $H_j$, precisions $w_i$ for $i\in [n]$, and a real
  $r>0$. Other parameters, $p,\eps,\rho,m$, are as in Alg.~\ref{alg:ellp-sketch}.}
\label{alg:ellp-reconstruction}
For each $i\in[n]$, compute 
$\hat x_i=\median_{j\in[l]}
   \big\{\tfrac{\left|H_j(h_j(i)) / r \right|^p}{w_i}\big\}$.
\\
Apply PSL reconstruction algorithm $R$ to vector $(\hat x_1,\ldots
\hat x_n)$ and $(w_1,\ldots w_n)$, and let $\hat \sigma$ be its
output. Explicitly, for each $i\in [n]$, if $\hat x_iw_i\ge
t\triangleq 4/\eps$, 
then set $s_i\triangleq\tfrac{1}{k}+\tfrac{k-1}{k}\cdot \tfrac{\hat
  x_iw_i/t-1}{w_i-1}$
(recall $k= \zeta\rho^{-1}\eps^{-2}$ from PSL), otherwise 
$s_i\triangleq 0$; then, let $\hat\sigma=t\sum_i s_i$.
\\
Output $r\cdot \hat \sigma$.
\end{algorithm}

\subsection{Estimating $F_k$ Moments for $k>2$}
\label{sec:FkMoments}

We now present the algorithm for estimating $F_k$ moments for
$k>2$, using the PSL Lemma~\ref{lem:nonUniformSampling}. To reduce the
clash of parameters, we refer to the problem as
``$F_p$ moment estimation''.

\newcommand{\knorm}{p}

\begin{theorem}
\label{thm:momentEstimation}
Fix $n\ge 8$, $p>2$, and $0<\eps<1/3$. There is a randomized linear function
$L:\R^n\to
\R^{S}$, with $S=O(n^{1-2/p}\cdot p^2\eps^{-2-4/\knorm}\log n)$, and a
deterministic estimation algorithm $E:\R^S\to \R$, 
such that for every $x\in \R^n$, 
with probability at least $0.51$, its output
$E(L(x))$ approximates $\|x\|_\knorm^\knorm$ within factor $1+\eps$.
\end{theorem}

\begin{proof}[Proof of Theorem~\ref{thm:momentEstimation}]
Our linear sketch $L$ is Alg.~\ref{alg:ellp-sketch}, and the
estimation algorithm $E$ is Alg.~\ref{alg:ellp-reconstruction}, with
the following choice of parameters.  Let
$\rho=\tfrac{\eps/4}{n^{p/2-1}}$. Let $\W=\W(k)$, for
$k=\zeta\rho^{-1}\eps^{-2}$, be from PSL
Lemma~\ref{lem:nonUniformSampling}.  Define
$\omega=9\EE{w\in\W}{w^{2/\knorm}}$, and note that $\omega\le
O(\rho^{-2/\knorm}\eps^{-4/\knorm})$ by
Lemma~\ref{lem:nonUniformSampling}.  Finally we set $m=\alpha\cdot
O(\rho^{-2/\knorm}\eps^{-4/\knorm})$ so that $m\ge \alpha\omega$,
where $\alpha=\alpha(\knorm,\eps)>1$ will be determined later.

In Alg.~\ref{alg:ellp-reconstruction}, we set $r$ to be a factor
$1-1/\knorm$ approximation to $\|x\|_2$, i.e., $(1-1/\knorm)\|x\|_2\le
r\le \|x\|_2$. Note that such $r$ is easy to compute (with high
probability) using, say, the AMS linear sketch \cite{AMS}, with
$O(p^2\log n)$ additional space. Thus, for the rest, we will just
assume that $\|x\|_2\in [1-1/\knorm,1]$ and set $r=1$.

The plan is to apply PSL Lemma~\ref{lem:nonUniformSampling} where
each unknown value $a_i$ is given by $|x_i|^p$, 
and each estimate $\hat a_i$ is given by $\hat x_i$. For this purpose, 
we need to prove that the $\hat x_i$'s are good approximators.
We thus let $F_2=\sum_{i=1}^n (x_iw_i^{1/\knorm})^2$. Note that
$\E{F_2}=\|x\|_2^2\cdot\EE{w\in\W}{w^{2/\knorm}}\le \omega/9$, and hence
by Markov's inequality, with probability at least $8/9$ 
we have $F_2\le \omega$.
\begin{claim}
Assume that $F_2\le \omega$. 
Then with high probability (say $\geq 1-1/n^2$) over the choice of the
hash tables, for every $i\in[n]$ 
the value $\hat x_i$ is a $(1/w_i,e^\eps)$-approximator to $|x_i|^\knorm$.
\end{claim}
\begin{proof}
We shall prove that for each $i\in[n]$ and $j\in[l]$,
with probability $\ge 8/9$ over the choice of $h_j$ and $g_j$,
the value $\tfrac{|H_j(h_j(i))|^\knorm}{w_i}$ is a 
$\left(1/w_i,e^\eps\right)$-approximator to $|x_i|^\knorm$.
Recall that each $\hat x_i$ is the median of $|H_j(h_j(i))|^\knorm/w_i$ over
$l=O(\log n)$ values of $j$, we get by applying a Chernoff bound
that with high probability it is a $(1/w_i,e^\eps)$-approximator to $|x_i|^p$.
The claim then follows by a union bound over all $i\in[n]$.

Fix $i\in[n]$ and $j\in[l]$, let $Y\triangleq H_j(h_j(i))$.
For $f\in[n]$, define $y_f=g_j(f)\cdot x_fw_f^{1/\knorm}$ if
$h_j(f)=h_j(i)$ and $0$ otherwise. Then
$Y=y_i+\delta$ where $\delta\triangleq \sum_{f\neq i} y_f$. Ideally,
we would like that $|Y|^p\approx 
|y_i|^p=|x_i|^pw_i$, i.e., the effect of the error $\delta$ is
small. Indeed, $\E{\delta^2}=\E{(\sum_{f\neq i}  
  y_f)^2}=\tfrac{1}{m}\sum_{f\neq i}(x_fw_f^{1/p})^2\le F_2/m$. 
Hence, by Markov's inequality,
$|\delta|\le \sqrt{9F_2/m}\le 3/\sqrt{\alpha}$
with probability at least $8/9$.

We now argue that if this event $|\delta|\le 3/\sqrt{\alpha}$ occurs, then 
$\tfrac{|H_j(h_j(i))|^\knorm}{w_i} 
 = \tfrac{|Y|^{\knorm}}{w_i}
 = \big|g_j(i)x_i+{\delta}/{w_i^{1/\knorm}}\big|^{\knorm}$ 
is a good approximator to $|x_i|^\knorm$.
Indeed, if $|\delta|/w_i^{1/\knorm} \le \tfrac{\eps}{2\knorm} |x_i|$, 
then clearly $\tfrac{|Y|^{\knorm}}{w_i}=(1\pm \tfrac{\eps}{2\knorm})^{\knorm}|x_i|^\knorm$. 
Otherwise,
since $|\delta|\le 3/\sqrt{\alpha}$, 
we have that
\begin{eqnarray*}
\left||Y|^p-|x_iw_i^{1/p}|^p\right|
&\le&
(|x_iw_i^{1/p}|+|\delta|)^p-|x_iw_i^{1/p}|^p
\\
&\le&
(\tfrac{2p}{\eps}|\delta|+|\delta|)^p-(\tfrac{2p}{\eps}|\delta|)^p
\\
&\le&
|\delta|^p\cdot (2p/\eps)^p\cdot \left((1+\tfrac{\eps}{2p})^p-1\right)
\\
&\le&
(6p)^p\cdot \eps^{1-p}/\alpha^{p/2}.
\end{eqnarray*}

If we set $\alpha=(6\knorm)^2/\eps^{2-2/p}$, then we obtain that
$\tfrac{|Y|^\knorm}{w_i}$ is a $(1/w_i,e^\eps)$-approximator to
$|x_i|^\knorm$, with probability at least $8/9$. We now take median
over $O(\log n)$ hash tables and apply a union bound to reach the
desired conclusion.
\end{proof}

We can now complete the proof of Theorem~\ref{thm:momentEstimation}.
Apply PSL (Lemma~\ref{lem:nonUniformSampling}) with
$a_i=|x_i|^p$ and $\hat a_i=\hat x_i$'s. 
By \Holder's inequality for $p/2$ and the normalization $r=1$, 
we have 
$\|x\|_\knorm^\knorm\ge \|x\|_2^\knorm/n^{\knorm/2-1}\ge \rho/\eps$, 
and thus additive error $\rho$
transforms to multiplicative error $1+\eps$.
\rnote{The PSL multiplicative error is $e^{2\eps}$,
and thus altogether, with probability at least $2/3-1/9-1/n^2\ge .51$,
the algorithm estimates $\|x\|_\knorm^\knorm$ within factor $e^{3\eps}$,
which matches theorem's assertion up to scaling $\eps$.}
\aanote{i agree. shall we add in text?}
It remains to bound the space: $S\le O(m\log n)
=O(\alpha \rho^{-2/\knorm}\eps^{-4/\knorm} \log n)
=O(\knorm^2/\eps^{2-2/p}\cdot \eps^{-6/\knorm}n^{1-2/\knorm}\cdot \log
n)=O(\knorm^2n^{1-2/\knorm}\cdot \eps^{-2-4/\knorm}\cdot \log n)$.
\end{proof}

\subsection{Estimating $\ell_1$ Norm}
\label{sec:ell1-norm}

To further illustrate the use of the Alg.~\ref{alg:ellp-sketch}
and~\ref{alg:ellp-reconstruction}, we now show how to use them for
estimating the $\ell_1$ norm.  In a later section, we obtain similar
results for all $\ell_p$, $p\in[1,2]$, except that the analysis is
more involved.

We obtain the following theorem. For clarity of presentation, the
efficiency (space and runtime bounds) are discussed separately below.

\begin{theorem}
\label{thm:ell1}
Fix $n\ge8$ and $8/n<\eps<1/8$. There is a randomized linear
function $L:\R^n\to \R^{S}$, with $S=O(\eps^{-3}\log^2 n)$, and a deterministic
estimation algorithm $E:\R^S\to \R$, such that for every $x\in \R^n$, 
with probability at least $0.51$, its output 
$E(L(x))$ approximates $\|x\|_1$ within factor $1+\eps$.
\end{theorem}
\begin{proof}
The sketch function $L$ is given by Alg.~\ref{alg:ellp-sketch}, with
parameters $p=1$, $\rho=\eps/8$, and $m=C\eps^{-3}\log n$ for a
constant $C>0$ defined shortly.  Let $\W=\W(k)$ for
$k=\zeta\rho^{-1}\eps^{-2}$ be obtained from the PSL
Lemma~\ref{lem:nonUniformSampling}.  Define $\omega=10\EE{w\in\W}{w\mid
  M}$, where event $M=M(w)$ satisfies $\Pr[M]\ge 1-O(n^{-2})$. Note
that $\omega\le O(\eps^{-3}\log n)$. We set constant $C$ such that $m\ge
3\omega$.

The estimation procedure is just several invocations of Alg.
\ref{alg:ellp-reconstruction} for different values of $r$.
For the time being, assume we hold an overestimate of $\|x\|_1$, 
which we call $r\ge\|x\|_1$. Then algorithm $E$ works by applying 
Alg.~\ref{alg:ellp-reconstruction} with this parameter $r$.

Let $F_1=\sum_{i=1}^n |x_iw_i|/r$. Note that $\E{F_1\mid \cap_i
  M(w_i)}=\|x\|_1/r\cdot\EE{w\in\W}{w\mid M(w)}\le \omega/10$, and
hence by Markov's inequality, $F_1\le \omega \le m/3$ with probability
at least $9/10-O(n/n^2)\ge 8/9$. Call this event ${\cal E}_r$, and
assume henceforth it indeed occurs.

To apply the PSL, we need to prove that each $\hat x_i$ in
Alg.~\ref{alg:ellp-reconstruction} is a good approximator to $x_i$. 
Fix $i\in[n]$ and $j\in[l]$. We claim that, conditioned on
${\cal E}_r$, the with probability at least $2/3$, 
$\tfrac{|H_j(h_j(i))|}{rw_i}$ is a $(1/w_i,1)$-approximator of $|x_i|$. 
Indeed,
$\tfrac{H_j(h_j(i))}{rw_i}=\tfrac{1}{r} g_j(i)x_i+\tfrac{1}{rw_i}\sum_{f\neq i:h_j(f)=h_j(i)}
g_j(f)w_fx_f$, and thus, 
$$
\E{\left|
\tfrac{|H_j(h_j(x))|}{rw_i}-\tfrac{|x_i|}{r}
\right|}
\le 
\tfrac{1}{rw_i}\sum_{f\neq i} \tfrac{1}{m}|x_f w_f|\le \tfrac{F_1}{m w_i}
\le \tfrac{1}{3w_i}.
$$
Hence, by Markov's inequality,
$\tfrac{|H_j(h_j(x))|}{rw_i}$ is a $(1/w_i,1)$-approximator of $|x_i|/r$
with probability at least $2/3$. 
By a Chernoff bound, their median $\hat x_i=\median_{j\in[l]}
\big\{\tfrac{|H_j(h_j(i))|}{rw_i}\big\}$ is
a $(1/w_i,1)$-approximator to $|x_i|/r$ with probability at least
$1-n^{-2}$. Taking a union bound over all $i\in[n]$ and applying
the PSL (Lemma~\ref{lem:nonUniformSampling}), we obtain that the PSL output,
$\hat\sigma=\hat\sigma(r)$ is an $(\eps/8, e^\eps)$-approximator to
$\|x\|_1/r$, with probability at least $2/3-1/9-1/n^2\ge 0.6$.
\rnote{Can we indeed PSL (say it succeeds w.p. 2/3) even after 
conditioning on $\cal E_r$,
which assumed a certain sum of $w_i$'s is bounded?}
\aanote{I think so: think of it as a event space of $w$'s product with
  distribution of hash functions, and the event space is independent in the two
  ``coordinates'' (product measure). There is a certain 
  part of the space, 1/3, where PSL fail, by failure of the first
  coordinate. In the rest of the space, of
  measure 2/3, we have good stuff happening, with probability
  $\ge8/9$, over the choice of the second coordinate. In total, we
  have $1-1/3-2/3\cdot (1/9)$ measure of the good event. does this
  sound right?
}

Now, if we had $r\le 4\|x\|_1$, then we would be done as $r\hat\sigma$ would
be a $(\eps\|x\|_1/2,e^\eps)$-approximator to $\|x\|_1$, and hence a
$1+2\eps$ multiplicative approximator
(and this easily transforms to factor $1+\eps$ by suitable scaling of $\eps$). 
Without such a good estimate
$r$, we try all possible values $r$ that are powers of $2$, from high to low, until we make
the right guess. Notice that it is easy to verify that the current guess 
$r$ is sufficiently large that we can safely decrease it.
Specifically, if $r>4\|x\|_1$ then
$r\hat\sigma<e^\eps\|x\|_1+\eps r/8\le (r/4)\cdot [1+3\eps/2+\eps/2] =
(1+2\eps)r/4$. 
However, if $r\le 2\|x\|_1$ then
$r\hat\sigma\ge e^{-\eps}\|x\|_1-\eps r/8\ge
(r/2)\cdot [1-\eps-\eps/4] > (1+2\eps)r/4$.
We also remark that, while we repeat 
Alg.~\ref{alg:ellp-reconstruction} for $O(\log n)$ times
(starting from $r=n^{O(1)}$ suffices), there is no
need to increase the probability of success as 
the relevant events ${\mathcal E}_r=\{\sum_i |x_iw_i| \le rm/3\}$ 
are nested and contain the last one, where $r/\|x\|_1\in[1,4]$.
\end{proof}

\subsection{The Running Times}

We now briefly discuss the runtimes of our algorithms: the update time
of the sketching Alg.~\ref{alg:ellp-sketch}, and the
reconstruction time of the Alg.~\ref{alg:ellp-reconstruction}.

It is immediate to note that the update time of our sketching
algorithm is $O(\log n)$: one just has to update $O(\log n)$ hash
tables. We also note that we can compute a particular $w_i$ in $O(\log
n)$ time, which is certainly doable as $w_i$ may be generated directly
from the seed used for the pairwise-independent
distribution. Furthermore, we note that we can sample from the
distribution $\W=\W(k)$ in $O(1)$ time (see, e.g.,
\cite{I10-consSampling}).  \aanote{not the best reference!}

Now we turn to the reconstruction time of
Alg.~\ref{alg:ellp-reconstruction}. As currently described, this
runtime is $O(n\log n)$. One can improve the runtime by using the
CountMin heavy hitters (HH) sketch of~\cite{CM05b}, at the cost of a
$O(\log(\tfrac{\log n}{\eps}))$ factor increase in the space and
update time. This improvement is best illustrated in the case of
$\ell_1$ estimation. We construct the new sketch by just applying the
$\Theta(t/m)$-HH sketch (Theorem 5 of~\cite{CM05b}) to the vector
$x\cdot w$ (entry-wise product). The HH procedure returns at most
$O(m/t)$ coordinates $i$, together with
$(1/w_i,e^{\eps})$-approximators $\hat x_i$, for which it is possible
that $\hat x_i w_i\ge t$ (note that, if the HH procedure does not
return some index $i$, we can consider 0 as being its
approximator). This is enough to run the estimation procedure $E$ from
PSL, which uses only $i$'s for which $\hat x_iw_i\ge t$. Using the
bounds from \cite{CM05b}, we obtain the following guarantees. The
total space is $O(\eps^{-1}\log n\log(\tfrac{\log n}{\eps})\cdot
m/t)=O(m\log n\cdot \log(\tfrac{\log n}{\eps}))=O(\eps^{-3}\log^2
n\cdot \log(\tfrac{\log n}{\eps}))$. The update time is $O(\log n\cdot
\log(\tfrac{\log n}{\eps}))$ and reconstruction time is $O(\log^2 n\cdot
\log(\tfrac{\log n}{\eps}))$.

To obtain a similar improvement in reconstruction time for the
$F_k$-moment problem, one uses an analogous approach, except that one
has to use HH with respect to the $\ell_2$ norm, instead of the
$\ell_1$ norm (considered in \cite{CM05b}).

\section{Applications II: Bounds via $p$-Type Constant}

In this section, we show further applications of the PSL to streaming
algorithms.  As in Section~\ref{sec:basicStreamingApps}, our sketching
algorithm will be linear, following the lines of the generic
Alg.~\ref{alg:ellp-sketch}.

An important ingredient for our intended applications will be a
variation of the notion of {\em $p$-type} of a Banach space (or, more
specifically, the $p$-type constant). This notion will give a
bound on the space usage of our algorithms, and hence we will bound it
in various settings. Below we state the simplest such bound,
which is a form of the Khintchine inequality.

\begin{lemma}
\label{lem:randomSum1D}
Fix $p\in[1,2]$, $n\ge 1$ and $x\in \R^n$. 
Suppose that for each $i\in[n]$ we have two random variables,
$g_i\in\{-1,+1\}$ chosen uniformly at random,
and $\chi_i\in\zo$ chosen to be $1$ with probability $\alpha\in (0,1)$ 
(and $0$ otherwise). Then 
$$
\textstyle
\E{\Big|\sum_i g_i\chi_ix_i\Big|^p}\le \alpha\|x\|_p^p.
$$

Furthermore, suppose each family of random variables $\{g_i\}_i$ and
$\{\chi_i\}_i$ is only pairwise independent and the two families are independent of each
other. Then, with probability at least 7/9, we have that
$$
\left|\sum_i g_i\chi_ix_i\right|^p\le 3^{2+p}\alpha\|x\|_p^p.
$$
\end{lemma}
\rnote{The dependence assumption are not stated clearly;
should be stated for for each of the sequences $g_i$ and $\chi_i$ separately, 
and then between them.
}\aanote{how about now?}

The proof of this lemma appears in Section~\ref{sec:pTypeProofs}.

\subsection{$\ell_p$-norm for $p\in[1,2]$}
\label{sec:ellp-norm}

We now use Alg.~\ref{alg:ellp-sketch}
and \ref{alg:ellp-reconstruction} to estimate the $\ell_p$ norm for
$p\in[1,2]$. We use Lemma~\ref{lem:randomSum1D} to bound the space usage.

\begin{theorem}
\label{thm:ellp}
Fix $p\in[1,2]$, $n\ge 6$, and $0<\eps<1/8$. There is a randomized linear
function $L:\R^n\to \R^{S}$, with $S=O(\eps^{-2-p}\log^2 n)$, and a deterministic
estimation algorithm $E$, such that for every $x\in \R^n$,
with probability at least $0.51$,
$E(L(x))$ is a factor $1+\eps$ approximation to $\|x\|_p^p$.
\end{theorem}
\rnote{If the proof is the same as $\ell_1$, without using these
$l_p$ inequalities, then why do the $\ell_1$ separately?}
\aanote{it's not a super-good reason, but the reason is that the analysis
  here is simpler. e.g., the +/-1 is not needed for $\ell_1$, but is
  needed for $\ell_p$.}

\begin{proof}
Our sketch function $L$ is given by Alg.  \ref{alg:ellp-sketch}. We
set $\rho=\eps/8$.  Let $\W=\W(k)$ for $k=\zeta\rho^{-1}\eps^{-2}$
obtained from the PSL (Lemma~\ref{lem:nonUniformSampling}).  Define
$\omega=10\EE{w\in\W}{w\mid M}$, where event $M=M(w)$ satisfies
$\Pr[M]\ge 1-O(n^{-2})$. Note that $\omega\le O(\eps^{-3}\log n)$. We
set $m=\alpha\omega$ for a constant $\alpha>0$ to be determined
later. \aanote{there seems to be 1/p exponent missing in some places.}

We now describe the exact reconstruction procedure, which will be
just several invocations of the algorithm 
\ref{alg:ellp-reconstruction} for different values of $r$.
As in Theorem~\ref{thm:ell1}, we guess $r>0$ starting from the highest 
possible value and halving it each time,
until we obtain a good estimate: $\|x\|_p\le r\le 4\|x\|_p$
(alternatively, one could prepare for
all possible $r$'s). To simplify the exposition, let us just assume 
in the sequel that $r=1$ and thus $1/4\le\|x\|_p\le 1$.

Let $F_p=\sum_{i=1}^n |x_i|^pw_i$. Note that $\E{F_p\mid \cap_i
  M(w_i)}=\|x\|_p^p\cdot\EE{w\in\W}{w\mid M(w)}\le \omega/10$, and
hence by Markov's inequality, $F_p\le \omega$ with probability at
least $8/9$. Call this event ${\cal E}$ and assume henceforth it
occurs.  To apply PSL, we need to prove that every $\hat x_i$ from
Alg.~\ref{alg:ellp-reconstruction} is a good approximator to $x_i$.
\begin{claim}
\label{clm:approximator_ellp}
Assume $F_p\le \omega$ and fix $i\in[n]$. If $\alpha\ge
3^{2+p}\eps^{1-p}$, then with high probability, $\hat x_i$ is a
$(1/w_i,e^\eps)$-approximator to $|x_i|^p$.
\end{claim}
\begin{proof}
Fix $j\in[l]$; we shall prove that $|H_j(h_j(i))|^p$ is
a $(1,1+\eps)$-approximator to $|x_i|^pw_i$, with probability at
least $2/3$. Then we would be done by Chernoff bound, 
as $\hat x_i$ is a median over $l=O(\log n)$ independent trials $j\in[l]$.

For $f\in[n]$, define $y_f=g_j(f)\cdot x_iw_i^{1/p}$ if $h_j(f)=h_j(i)$ and
$y_f=0$ otherwise. Define $Y\triangleq H_j(h_j(i))=y_i+\delta$, where
$\delta=\sum_{f\neq i} y_f$. We apply Lemma~\ref{lem:randomSum1D} to conclude
that $\E{|\delta|^p}\le F_p/m$, and hence $|\delta|^p\le 3\omega/m\le
3/\alpha$ with probability at least $2/3$. 
Assume henceforth this is indeed the case.

Now we distinguish two cases. First, suppose $|x_iw_i^{1/p}|\ge
\tfrac{2}{\eps}\cdot |\delta|$. Then $|Y|^p = 
(1\pm\eps/2)|x_i|^pw_i$. Otherwise, $|x_iw_i^{1/p}|<
\tfrac{2}{\eps}\cdot |\delta|$, and then
\begin{eqnarray*}
\left||Y|^p-|x_iw_i^{1/p}|^p\right|
&\le&
(|x_iw_i^{1/p}|+|\delta|)^p-|x_iw_i^{1/p}|^p
\\
&\le&
|\delta|^p\cdot \left((2/\eps+1)^p-2/\eps\right)
\\
&\le&
|\delta|^p\cdot (2/\eps)^p\cdot (1+p\eps-1)
\\
&\le&
p2^p\cdot3\cdot \eps^{1-p}/\alpha.
\end{eqnarray*}

\aanote{updated the calculations to reflect the tight(er) bound.}

Thus, if we set $\alpha\ge 3^{2+p}(1/\eps)^{p-1}$, 
then in both cases $|Y|^p$ is a $(1,e^\eps)$-approximator to $|x_i|^pw_i$
(under the event that occurs with probability at least $2/3$). 
\end{proof}

We can now complete the proof of Theorem~\ref{thm:ellp}.
Applying Lemma~\ref{lem:nonUniformSampling}, we obtain that its output,
$\hat\sigma=\hat\sigma(r)$, is a $(\eps/8, e^{2\eps})$-approximator to
$\|x\|_p$, with probability at least $2/3-1/9-1/n^2\ge 0.51$.
\end{proof}

\subsection{Mixed and cascaded norms}
\label{sec:mixednorms}

We now show how to estimate mixed norms such as the $\ell_{p,q}$
norms. In the latter case, the input is a matrix $x\in \R^{n_1\cdot
  n_2}$, and the $\ell_{p,q}$ norm is  $\|x\|_{p,q}=(\sum_i
\|x_i\|_q^p)^{1/p}$, where $x_i$ is the $i$th row in the matrix. 

We show a more general theorem, for the norm $\ell_p(X)$,
which is defined similarly for a general Banach space $X$; 
the $\ell_{p,q}$ norms will be just
particular cases. To state the general result, we need the following
definition.

\begin{definition}
\label{def:generalizedType}
Fix $p\ge 1$, $n,\kappa\in \N$, $\omega>0$, $\delta\in[0,1)$, and let $X$ be a finite
dimensional Banach space. The 
{\em the generalized $p$-type}, denoted $\alpha(X,p,n,\kappa,\omega,\delta)$, is
the biggest constant $\alpha>0$ satisfying the following: 
For each $i\in[n]$, 
let $g_i\in \{-1,+1\}$ be a random variable drawn uniformly at random, 
and let $\chi_i\in\zo$ be a random variable that is equal $1$ with probability
$1/\alpha$ and $0$ otherwise. Furthermore, each family $\{g_i\}_i$ and
$\{\chi\}_i$ is $\kappa$-wise independent, and the two families are
independent of each other.
Then, for every $x_1,\ldots x_n\in X$ satisfying $\sum_{i\in[n]}
\|x_i\|_X^p\le \omega$, 
$$
\textstyle
  \Pr\left[ \big\|\sum_{i\in[n]} g_{i} \chi_ix_i\big\|_X^p \le 1 \right] 
  \ge 1-\delta.
$$
\end{definition}

\begin{theorem}
\label{thm:ellp-product}
Fix $p\ge 1$, $n\ge2$, and $0<\eps<1/3$. 
Let $X$ be a Banach space admitting a linear sketch $L_X:X\to \R^{S_X}$, 
with space $S_X=S_X(\eps)$, 
and let $E_X:\R^{S_X}\to \R$ be its reconstruction procedure. 
\rnote{$L_X$ does not depend on $\eps$? 
Actually both $L_X$ and $E_X$ should depend on $\eps'$ which is generally
different from $\eps$ and actually only a single value $\eps'=\eps/2$}
\aanote{yes, $L_X$ does depend on $\eps$, as was tried to be conveyed
  via $S_X$. shall we make this more explicit?}

Then there is a randomized linear
function $L:X^n\to \R^{S}$, and an
estimation algorithm $E$ which, for any $x\in X^n$,
given the sketch $Lx$, outputs a factor $1+\eps$
approximation to $\|x\|_{p,X}$,  with probability at least $0.51$.

Furthermore, $S\le S_X(\eps/2)\cdot \alpha(X,p,n,\kappa,O(p\eps^{-4}\log
n),2/3)\cdot O(\log n)$,
where $\kappa$ is such that each function $g_j$ and $h_j$ is $\kappa$-wise
independent.
\end{theorem}

We note that the result for $\ell_{p,q}$ norms will follow by proving
some particular bounds on the parameter $\alpha$, the generalized
$p$-type. We discuss
these implications after the proof of the theorem.

\begin{proof}[Proof of Theorem~\ref{thm:ellp-product}]
Our sketch function $L$ is given by algorithm
\ref{alg:ellp-sketch}, with one notable modification. $x_i$'s are now
vectors from $X$ and the hash table cells hold sketches given by
sketching function $L_X$ up to $1+\eps/2$ approximation. In
particular, each cell of hash table $H_j(z)=\sum_{i:h_j(i)=z}
g_j(i)\cdot w_i^{1/p}\cdot L_Xx_i$. Furthermore, abusing
notation, we use the notation $\|H_j(z)\|_q$ for some $z\in[m]$ to
mean the result of the $E$-estimation algorithm on the sketch
$H_j(z)$ (since it is a $1+\eps/2$ approximation, we can afford such
additional multiplicative error).

We set $\rho=\eps/8$. 
Let $\W=\W(k)$ by for $k=\zeta\rho^{-1}\eps^{-2}$ obtained from the PSL
Lemma~\ref{lem:nonUniformSampling}.
Define $\omega=10\EE{w\in\W}{w\mid M}$, where event $M=M(w)$ satisfies $\Pr[M]\ge 1-O(n^{-2})$. Note that $\omega\le
O(\eps^{-3}\log n)$. We set $m$ later.

We now describe the exact reconstruction procedure, which will be
just several invocations of the algorithm 
\ref{alg:ellp-reconstruction} for different values of $r$.
As in Theorem~\ref{thm:ell1}, we guess $r$ starting from high and halving it each time,
until we obtain a good estimate --- $\|x\|_{p,X}\le r\le 4\|x\|_{p,X}$
(alternatively, one could prepare for
all possible $r$'s). For simplified exposition, we just assume that
$1/4\le\|x\|_{p,X}\le 1$ and $r=1$ in the rest.

Let $F_{p,X}=\sum_{i=1}^n \|x_iw_i^{1/p}\|_X^p$. Note that
$\E{F_{p,X}\mid \cap M(w_i)}=\|x\|_X^p\cdot\EE{w\in\W}{w\mid M(w)}\le
\omega/10$, and hence $F_{p,X}\le \omega$ with probability at least
$8/9$ by Markov's bound. Call this event ${\cal E}$.  To apply PSL, we
need to prove that $\hat x_i$'s from
Alg.~\ref{alg:ellp-reconstruction} are faithful approximators. For
this, we prove that, for appropriate choice of
$\alpha=\alpha(p,X,\eps,n)$, for each $j\in[l]$, $\|H_j(h_j(i))\|_X^p$
is a $(1,1+\eps)$-approximator to $\|x_i\|_X^pw_i$, with probability
at least $2/3$. This would imply that, since $\hat x_i$ is a median
over $O(\log n)$ independent trials, $\hat x_i$ is a
$(1/w_i,1+\eps)$-approximator to $\|x_i\|_X^p$.  Once we have such a
claim, we apply Lemma~\ref{lem:nonUniformSampling}, and conclude that
the output, $\hat\sigma=\hat\sigma(r)$, is a $(\eps/8,
1+2\eps)$-approximator to $\|x\|_{p,X}$, with probability at least
$2/3-1/9-1/n\ge 0.51$.

\begin{claim}
Fix $p\ge1$ and $\omega\in \R_+$. Let
$m=\alpha(X,p,\kappa,3p\omega/\eps,2/3)$, the 
generalized $p$-type of $X$.

Assume $F_{p,X}\le \omega$ and fix $i\in[n],j\in[l]$. Then $\|H_j(h_j(i))\|_X^p$ is a
$(1,1+\eps)$-approximator to $\|x_i\|_X^p w_i$ with probability at least 2/3.
\end{claim}
\begin{proof}
For $f\in[n]$, define $y_f=g_j(f)\cdot x_iw^{1/p}$ if $h_j(f)=h_j(i)$ and
$y_f=0$ otherwise. Then, $a\triangleq
\sum_{f\in[n]:h_j(f)=h_j(i)}g_j(i)x_i=y_i+\delta$, where 
$\delta=\sum_{f\neq i} y_f$. Then, by the definition of generalized
$p$-type of $X$, whenever $m\ge
\alpha(X,p,\kappa,\omega\cdot\tfrac{3p}{\eps},2/3)$, we 
have that $\|\delta\|_X\le \eps/3$, with probability at least $2/3$.

Now we distinguish two cases. First, suppose $\|x_iw_i^{1/p}\|_X\ge
\tfrac{2p}{\eps}\cdot \|\delta\|_X$. Then $\|a\|_X^p\approx
(1\pm\eps)\|x_i\|_X^pw_i$. Otherwise, if $\|x_iw_i^{1/p}\|_X<
\tfrac{2p}{\eps}\cdot \|\delta\|_X$, then
$$
\|a\|_X^p\le \left(\|x_iw_i^{1/p}\|_X+\|\delta\|_X\right)^p\le (2p\|\delta\|_X/\eps+\|\delta\|_X)^p\le
\|\delta\|_X^p\cdot(2p/\eps+1)^p\le 1. 
$$
Hence, we conclude that $\|a\|_X^p$ (and thus $\|H_j(h_j(i))\|_X^p$) is a
$(1,1+\eps)$-approximator to $\|x_i\|_X^pw_i$, with probability at least $2/3$.
\end{proof}
The claim conclude the proof of Theorem~\ref{thm:ellp-product}. Note
that the space is $S=O(S_X(\eps/2)\cdot \alpha(X,p,\kappa,O(p\eps^{-4}\log
n),2/3)\cdot \log n)$.
\end{proof}

We now show the implications of the above theorem. For this, we present
the following lemma, whose proof is included in Section~\ref{sec:pTypeProofs}.

\begin{lemma}
\label{lem:pTypeHighD}
Fix $n,m\in \N$, $\omega\in \R_+$, and a finite dimensional Banach space $X$. 
We have the following bounds on the generalized $p$-type:
\begin{enumerate}
\renewcommand{\theenumi}{(\alph{enumi})}
\item
\label{it:1pq2}
if $0< p\le q\le 2$, then
$\alpha(\ell_q^{m},p,n,2,\omega,2/3) \le O(\omega)$.
\item
\label{it:pq2}
if $p,q\ge 2$, 
we have that
$\alpha(\ell_q^{m},p,n,2q,\omega,2/3) \le 9^2q^{O(1)}\omega^{2/p}\cdot
n^{1-2/p}$,
and if $q\ge 2$ and $p\in(0,2)$, then
$\alpha(\ell_q^{m},p,n,2q,\omega,2/3) \le 9^2q^{O(1)}\omega^{2/p}$.
\item
\label{it:Xp}
for $p\ge 1$, we have that
$\alpha(X,p,n,2,\omega,2/3) \le O(n^{1-1/p}\omega^{1/p})$,
and for $p\in(0,1)$, we have that
$\alpha(X,p,n,2,\omega,2/3) \le O(\omega^{1/p})$.
\end{enumerate}
\end{lemma}

Combining Theorem~\ref{thm:ellp-product} and
Lemma~\ref{lem:pTypeHighD}, also using
Theorem~\ref{thm:momentEstimation}, we obtain the following linear sketches
for $\ell_{p,q}$ norms, which are optimal up to $(\eps^{-1}\log
n)^{O(1)}$ factors (see, e.g., \cite{JW-cascaded}\aanote{verify}).

\begin{corollary}
\label{cor:mixed-ellpq}
There exist linear sketches for $\ell_p^{n_1}(\ell_q^{n_2})$,
for $n_1,n_2\le n$ and $p,q\ge 1$, with the following space bounds $S$.

For $0< p\le q\le 2$, the bound is $S=(\eps^{-1}\log n)^{O(1)}$.

For $q\ge 2$ and $p\in(0,2)$, the bound is $S=n_2^{1-2/q}\cdot
(pq\eps^{-1}\log n)^{O(1)}$. 

For $p,q\ge 2$, the bound is $S=n_1^{1-2/p}n_2^{1-2/q}\cdot
(pq\eps^{-1}\log n)^{O(1)}$. 

For $p\ge 1$ and $q\in(0,p)$, the bound is $S=n_1^{1-1/p}\cdot
(\eps^{-1}\log n)^{O(1)}$. 

For $p\in(0,1)$ and $q\in(0,p)$, the bound is $S=(\eps^{-1}\log n)^{O(1)}$.
\end{corollary}
\rnote{I did not check it; what about comparison with JW?}
\aanote{add all the cases: $q\in[0,2]$ and $p>q$; $q>2$ and $p\in[0,2]$. the case
  of $p>q$ also works for $p>0$}

\section{Applications III: Sampling from the Stream}
\label{sec:ellp-sampling}

We now switch to a streaming application of a different type,
$\ell_p$-sampling, where $p\in[1,2]$. We obtain the following theorem.

\begin{theorem}
\label{thm:ellp-sampling}
Fix $n\ge2$, $p\in[1,2]$, and $0<\eps<1/3$. There is a randomized linear
function $L:\R^n\to \R^{S}$, with $S=O(\eps^{-p}\log^3 n)$, and an
``$\ell_p$-sampling algorithm $A$'' satisfying the following. For any
non-zero $x\in \R^n$, there is a distribution $D_x$ on $[n]$ such that $D_x(i)$
is a $(n^{-2},1+\eps)$-approximator to $|x_i|^p/\|x\|_p^p$. Then $A$
generates a pair $(i, v)$ such that $i$ is drawn from $D_x$ (using the
randomness of the function $L$ only), and $v$ is a
$(0,1+\eps)$-approximator to $|x_i|^p$. 
\end{theorem}

In this setting, the sketch algorithm is essentially the
Algorithm~\ref{alg:ellp-sketch}, with the following minor modification.
We use $k=\zeta t\cdot \log n$ for a sufficiently high $\zeta>0$, and
choose $m=O(k\eps^{-p}\log n)=O(\eps^{-1-p}\log^2n)$ (note that the
choice of $\rho$ is irrelevant as it affects only parameter $k$, fixed
directly).  Furthermore, the algorithm is made to use limited
independence by choosing $w_i$'s as follows. Fix $k$ seeds for
pair-wise independent distribution. Use each seed to generate the list
$\{w_{i,j}\}_{j\in[n]}$, where each $w_{i,j}=1/u_{i,j}$ for random
$u_{i,j}\in U(0,1)$. Then $w_i=\max_{j\in[k]} w_{i,j}$ for each
$i\in[n]$. Note that each $w_i$ has distribution $\W=\W(k)$. This
method of generating $w_i$'s leads to an update time of $O(k+\log
n)=O(\eps^{-1}\log n)$.

Given the sketch, the sampling algorithm proceeds as described in
Alg.~\ref{alg:ellp-sampling} (using $w_{i,j}$'s defined above). We set
$r$ to be a $2$
approximation to $\|x\|_p^p$, which is easy to compute
separately (see, e.g., Theorem~\ref{thm:ellp}).  So, below we
just assume that $1/2\le \|x\|_p^p\le 1$ and $r=1$.

\begin{algorithm}[!h]
\caption{$\ell_p$-sampling algorithm. Input consists
  of $l$ hash tables $H_j$, precisions $w_{i,j}$ for $i\in[n], j\in[k]$, and a real
  $r>0$.}
\label{alg:ellp-sampling}
Compute 
$\hat x_i=\median_{j=1\ldots
  l}\left\{\left|\tfrac{H_j(h_j(i))\ /\ r}{w_i}\right|^p\right\}$,
where $w_i=\max_{j\in[k]} w_{i,j}$.
\\
We compute the following quantities $s_{i,j}\in\zo$ for $i\in[n]$ and
$j\in[k]$. 
For each $i\in[n], j\in[k]$, let $s_{i,j}=1$ if $\hat x_iw_{i,j}\ge
t\triangleq 4/\eps$ and 0 otherwise.
\\
Let $j^*$ be the smallest $j\in[k]$ such that there is exactly one
$i\in[n]$ with $s_{i,j^*}=1$.
\\
If such $j^*$ exists, 
return $(i^*,\hat x_{i^*}\cdot r^p/t)$ where $i^*$ is the unique $i^*$ with $s_{i^*,j^*}=1$.
\\
If no $j^*$ exists, return FAIL.
\end{algorithm}

\begin{proof}[Proof of Theorem~\ref{thm:ellp-sampling}]
Let $\omega=10\EE{w\in\W}{w\mid M}=O(k\log n)$, where event $M=M(w)$
satisfies $\Pr[M]\ge 1-\Omega(n^{-2})$. We choose the constant in
front of $m$ such that $m\ge \alpha\omega$ for
$\alpha=3^{2+p}\eps^{1-p}$.

Define $F_p=\sum_{i\in[n]} (x_iw_i^{1/p})^{p}$. Note that $\E{F_p\mid
  \cap_i M(w_i)}=\|x\|_p^p\cdot \omega/10$. Hence $F_p\le \omega$ with
probability at least $9/10-O(1/n)\ge 8/9$. By
Claim~\ref{clm:approximator_ellp}, we deduce that $\hat x_i$ is a
$(1/w_i,e^\eps)$-approximator to $|x_i|^p$, with high probability.

We now prove that the reconstruction algorithm samples an element $i$
with the desired distribution. We cannot apply PSL black-box anymore,
but we will reuse of the ingredients of PSL below. Let
$a_i=|x_i|^p\in[0,1]$, and $\hat a_i=\hat x_i$. Note that $\sum_i a_i\in[1/2,1]$.

The proof of correctness follows the outlines of the PSL proof. We
bound the probability that $s_{i,j}=1$, for $i\in[n], j\in[k]$ as follows:
$$
1/t\cdot a_i e^{-3\eps/2}
\le
\Pr[a_i\ge t/w_i+1/w_i]
\le
\Pr[s_{i,j}=1]
\le 
\Pr[a_i\ge t/w_i-1/w_i]
\le 1/t\cdot  a_ie^{3\eps/2}.
$$
Hence, for fixed $j$, $\sum_i t\cdot \Pr[s_{i,j}= 1]\le
1/t\cdot  e^{3\eps/2}\sum_i a_i\le
1/t\cdot  e^{3\eps/2}\le \eps/2$. Then, using pairwise
independence, for fixed $i,j$, we have that $s_{i,j}=1$ while all the
other $s_{i',j}=0$ for $i'\in[n]\setminus\{i\}$ with probability that satisfies
\begin{equation}
\label{eqn:PrIchosen}
1/t\cdot a_i e^{-3\eps/2}
\cdot
(1-\eps/2)
\le
\Pr[s_{i,j}=1 \wedge \sum_{i'\neq i}s_{i',j}=0]
\le 1/t\cdot a_i e^{3\eps/2}.
\end{equation}
Thus, $\sum_i s_{i,j}=1$ with probability at least
$\Omega(\eps)$. Furthermore, since the events for different $j\in[k]$
for $k=O(\eps^{-1}\log n)$ are
independent, the algorithm is guaranteed to not fail (i.e., reach step
5) with high probability.

It remains to prove that $i^*$ is chosen from some distribution $D_x$,
such that $D_x(i)$ is a $(O(n^{-2}),1+O(\eps))$-approximator to
$|x_i|^p/\|x_i\|_p^p$. Indeed, consider $j=j^*$, i.e., condition on the fact that
$\sum_{i} s_{i,j}=1$. Then,
$$
\Pr[i=i^*]=\frac{\Pr[s_{i^*,j}=1 \wedge \sum_{i'\neq i^*}s_{i',j}=0]}
{\sum_i \Pr[s_{i,j}=1 \wedge \sum_{i'\neq i}s_{i',j}=0]},
$$ which, by Eqn.~\ref{eqn:PrIchosen}, is a
$(n^{-2},e^{O(\eps)})$-approximator to $|x_i|^p/\|x\|_p^p$ (the
$O(n^{-2})$ terms comes from conditioning on event $M$). Also, note
that, since $s_{i^*,j}=1$, we have that $\hat{x}_i$ is a
$(0,e^{O(\eps)})$-approximator to $|x_i|^p$.

Scaling $\eps$ appropriately gives the
claimed conclusion.

The space bound is $S=O(m\log n)=O(\eps^{1-p}k\log^2
n)=O(\eps^{-p}\log^3 n)$.

\end{proof}

\section{Proofs of $p$-type inequalities}
\label{sec:pTypeProofs}

\begin{proof}[Proof of Lemma~\ref{lem:randomSum1D}]
Let us denote $g=(g_1,\ldots g_n)$ and $\chi=(\chi_1,\ldots \chi_n)$. 
Since $z^{p/2}$ is concave for $p\le 2$, 
a random variable $Z\ge 0$ satisfies $\E{Z^{p/2}} \le (\Enop{Z})^{p/2}$,
and thus 
\begin{eqnarray*}
\textstyle
\EE{g,\chi}{\Big|\sum_i g_i\chi_ix_i\Big|^p}
&=&
\textstyle
\EE{\chi}{\EE{g}{\left(\Big|\sum_i g_i\chi_ix_i\Big|^2\right)^{p/2}}}
\\
&\le&
\textstyle
\EE{\chi}{\left(\EE{g}{\Big|\sum_i g_i\chi_ix_i\Big|^2}\right)^{p/2}}.
\end{eqnarray*}
Now using (pairwise) independence of the sequence $g_1,\ldots,g_n$,
and the fact that $\|z\|_2\le \|z\|_p$, we conclude that
\begin{eqnarray*}
\textstyle
\EE{g,\chi}{\Big|\sum_i g_i\chi_ix_i\Big|^p}
&\le&
\textstyle
\EE{\chi}{\Big(\sum_i (\chi_ix_i)^2\Big)^{p/2}}
\\
&\le&
\textstyle
\EE{\chi}{\Big.\Big.\sum_i |\chi_ix_i|^p}
\\
&=&
\alpha\|x\|_p^p.
\end{eqnarray*}

We proceed to prove the lemma's second assertion.
Since $g$ and $\chi$ are independent,
for every fixed $\chi$ we have, by Markov's inequality, 
that with probability at least $8/9$ (over the choice of $g$),
$$
\Big|\sum_i g_i\chi_ix_i\Big|^2
\le 9\sum_i (\chi_ix_i)^2.
$$
Call vector $g$ satisfying the above ``good'' for the given vector $\chi$.
We henceforth restrict attention only to $g$ that is indeed good 
(for the relevant $\chi$, which is now a random variable), and we get
\begin{eqnarray*}
\textstyle
\EE{\chi_i}{\Big|\sum_i g_i\chi_ix_i\Big|^p}
&\le&
\textstyle
\EE{\chi_i}{\Big(9\sum_i (\chi_ix_i)^2\Big)^{p/2}}
\\
&\le&
3^p\EE{\chi_i}{\sum_i (\chi_ix_i)^p}
\\
&=&
3^p\alpha\|x\|_p^p,
\end{eqnarray*}
where the last inequality used again the fact that $\|z\|_2\le \|z\|_p$.
Now using Markov's inequality over the choice of $\chi$,
with probability at least $8/9$ we have 
$\left|\sum g_i\chi_ix_i\right|^p\le 3^{2+p}\alpha \|x\|_p^p$. 
The lemma now follows by recalling that $\chi$ and $g$ are independent
(or a union bound).
\end{proof}

\begin{proof}[Proof of Lemma~\ref{lem:pTypeHighD}]
For part~\ref{it:1pq2}, suppose that $0< p\le q\le 2$. We note that:
\begin{equation}
\label{eqn:normPQ}
\left\|\sum_i g_i\chi_ix_i\right\|_q^p
=
\left(\sum_{j} \left|\sum_i g_i\chi_ix_{ij}\right|^q\right)^{p/q}.
\end{equation}
We want to bound $\sigma(\chi,g)=\sum_{j} \left|\sum_i g_i\chi_ix_{ij}\right|^q$, for
fixed vector $\chi$ and random vector $g$. For fixed $j$, we have that, using
concavity of $x^{q/2}$, pairwise-independence, and norm-inequality respectively:
$$
\EE{g}{\left|\sum_i g_i\chi_ix_{ij}\right|^q}
\le
\left|\EE{g}{\left(\sum_i g_i\chi_ix_{ij}\right)^2}\right|^{q/2}
=
\left|\sum_i (\chi_ix_{ij})^2\right|^{q/2}
\le
\sum_i \chi_i|x_{ij}|^q.
$$
By linearity of expectation, $\EE{g}{\sigma(\chi,g)}\le \sum_i\chi_i \sum_j |x_{ij}|^q$. By
Markov's bound, we have that $\sigma(\chi,g)\le 9
\sum_i\|\chi_ix_i\|_q^q$, with probability at least
$8/9$ (over the choice of $g$). Call such $g$ good. Plugging this into
Eqn.~\eqref{eqn:normPQ}, since $p$-norm upper bounds $q$-norm, we have that:
$$
\left\|\sum_i g_i\chi_ix_i\right\|_q^p\le 9\cdot \sum_i \|\chi_ix_i\|_q^p.
$$
Conditioned on good $g$, by taking the expectation over $\chi_i$'s and
using Markov's bound, we obtain that
$$
\left\|\sum_i g_i\chi_ix_i\right\|_q^p\le 
\tfrac{9}{\alpha}\cdot 9\|x\|_{p,q}^p
$$
with probability at least $8/9$ over the choice of $\chi$.
Hence, $\left\|\sum_i g_i\chi_ix_i\right\|_q^p\le 1$ as long as
$\alpha=9^{2}\|x\|_{p,q}^p\le 9^{2}\omega$, with probability at
least $7/9$ over the choice of $g$ and $\chi$.

For part~\ref{it:pq2}, suppose that $q\ge 2$. As before, since
$$
\left\|\sum_i g_i\chi_ix_i\right\|_q^2
=
\left(\sum_{j} \left|\sum_i g_i\chi_ix_{ij}\right|^q\right)^{2/q},
$$
we want to bound $\sigma(\chi,g)=\sum_j \sigma_j(\chi,g)$,
where $\sigma_j(\chi,g)=\left|\sum_i \chi_i g_i x_{ij}\right|^q$. For
fixed $\chi$, we
compute the expectation $\EE{g}{\sigma_j(\chi,g)}$. For this we compute
the moment $\kappa=2\lceil q/2\rceil$ of $|\sum_i
g_i\chi_ix_{ij}|$. For convenience, define $y_i=\chi_ix_{i,j}$. We have that
$$
M_{\kappa}
\triangleq
\EE{g}{\left(\sum_i g_iy_i\right)^{\kappa}}
\le
\kappa!\cdot \left(\sum_i y_i^2\right)^{\kappa/2}.
$$
Hence, by concavity of $f(z)=z^{q/\kappa}$, we have
$$
\EE{g}{\sigma_j(\chi,g)}\le \left(M_\kappa\right)^{q/\kappa}\le (\kappa!)^{q/\kappa}\cdot
\left(\sum_i y_i^2\right)^{q/2}=q^{O(q)}\left\|\sum_i\chi_ix_{ij}\right\|_2^{q/2}.
$$
Thus, we have that $\sigma(\chi,g)\le 9q^{O(q)}{}\sum_j\|\sum_i\chi_i x_{ij}\|_2^{q/2}$ with
probability at least $8/9$. Again call such $g$'s good. For such a
good $g$, we now have that, by triangle inequality (in norm $q/2$):
$$
\left\|\sum_i g_i\chi_ix_i\right\|_q^2
\le 
9q^{O(1)}\cdot \left(\sum_j \left(\sum_i (\chi_i
    x_{ij})^2\right)^{q/2}\right)^{2/q}
\le
9q^{O(1)}\cdot \sum_i \|\chi_i x_{i}\|_q^{2}.
$$
Conditioned on good $g$, again by taking expectation over $\chi$ and
using Markov's bound, we obtain that, with probability at least $8/9$,
$$
\left\|\sum_i g_i\chi_ix_i\right\|_q^2
\le 9^2\cdot q^{O(1)}\cdot \tfrac{1}{\alpha} \cdot\sum_i\|x_i\|_q^2.
$$

Finally, we distinguish the cases where $p\ge 2$ and where $p\in(0,1)$. If $p\ge
2$, then using that $\sum_i \|x_i\|_q^2\le n^{1-2/p}\cdot
\|x\|_{p,q}^{2}\le n^{1-2/p}\omega^{2/p}$, we conclude that, with probability at least $7/9$
over $g,\chi$, we have that
$
\left\|\sum_i g_i\chi_ix_i\right\|_q^2\le 1
$
as long as $\alpha\ge 9^2q^{O(1)}n^{1-2/p}\omega^{2/p}$. Similarly, if
$p\in(0,2)$, then 
$\sum_i \|x_i\|_q^2\le 
\|x\|_{p,q}^{2}\le\omega^{2/p}$, and we conclude that, with probability at least $7/9$
over $g,\chi$, we have that
$
\left\|\sum_i g_i\chi_ix_i\right\|_q^2\le 1.
$
We note that we just used $\kappa$-wise independence, where $\kappa\le q+2$.

We now prove part~\ref{it:Xp}, which just follows from a triangle
inequality. Namely, we observe that
$$
\left\|\sum_i\chi_i g_ix_i\right\|_X
\le 
\sum_i\chi_i\left\|x_i\right\|_X.
$$
Hence, taking expectation and applying Markov's bound, we obtain, with
probability at least $8/9$, the following. If $p\ge 1$, then
$$
\left\|\sum_i\chi_i g_ix_i\right\|_X
\le 
\tfrac{9}{\alpha}\left\|x\right\|_{1,X}
\le 
\tfrac{9n^{1-1/p}}{\alpha}\left\|x\right\|_{p,X}
\le 
\tfrac{9n^{1-1/p}}{\alpha}\omega^{1/p},
$$
and taking $\alpha\ge 9n^{1-1/p}\omega^{1/p}$ is then enough. If
$p\in(0,1)$, then
$$
\left\|\sum_i\chi_i g_ix_i\right\|_X
\le 
\tfrac{9}{\alpha}\left\|x\right\|_{1,X}
\le 
\tfrac{9}{\alpha}\omega^{1/p},
$$
and taking $\alpha\ge 9\omega^{1/p}$ is enough.
\end{proof}

\section{A lower bound on the total precision}
\label{sec:lowerBound}

We now deduce a lower bound on $\sum_i \EE{}{w_i}$, and show it is
close to the upper bound that we obtain in
Lemma~\ref{lem:nonUniformSampling}.

\begin{theorem}
\label{thm:lowerBound}
Consider the same setting as in the Precision Sampling Lemma (Lemma~\ref{lem:nonUniformSampling}).
Let $\{a\}$ be a sequence of numbers in $[0,1]$. Let $\{w_i\}_{i \in [n]}$ be a sequence generated by a random process, independent of the sequence $\{a_i\}$. Let $R$ be an algorithm with the following properties.
The algorithm obtains both $\{w_i\}$ and a sequence $\{\hat  a_i\}_{i \in [n]}$, where each $\hat a_i$ is an arbitrary $(1/w_i,1)$-approximator to $a_i$. The algorithm outputs a value $\hat \sigma$ that is a $(\rho,e^\eps)$-approximator to $\sigma \eqdef \sum_{i \in [n]} a_i$ with probability at least $2/3$.

Let $\alpha = \max\{\rho/\eps,(6\eps)^{-4}\}$.
If $\eps \in (0,1/48)$, and $\alpha \le n/16$
then there exists an absolute positive constant $C$ such that
$\frac{1}{n}\sum_{i \in [n]} \mathbb E[w_i] \ge \frac{C}{\eps\rho} \cdot \log{(n/\alpha)}$.
\end{theorem}

Note that our lower bound is essentially off by a factor of $\eps$
from PSL. 

We now prove the theorem.
We start by adapting the lemma that shows that the Hoeffding bound is
nearly optimal.

\begin{lemma}[Based on Theorem~1 of \cite{CEG95}]\label{lem:lb_worst_query}
Let $\eps \in (0,1/8)$. Let $f$ be a function from $[n]$ to $\zo$. 
Let $t$ be a positive integer such that $t \le \sqrt{n/3} -1$.
Let $\mathcal A$ be a randomized algorithm 
that always queries the value of $f$ on at most $t$ different inputs, and outputs an estimate $\bar\sigma$ to $\sigma \eqdef \frac{1}{n} \cdot \sum_{x \in [n]} f(x)$.

If $|\bar\sigma - \sigma| < \eps$ with probability at least $7/12$, then $t \ge C/\eps^2$, where $C$ is a fixed positive constant.
\end{lemma}

\begin{proof}
Let $\delta$ be a bound on the probability that the algorithm returns an incorrect estimate.
In the proof of Theorem~1 in \cite{CEG95}, it is shown that 
$$\delta \ge \sum_{i=0}^{\lceil t/2 \rceil - 1} \binom{t}{i} \cdot \frac{\binom{n - t}{\lceil n(1/2+\eps)\rceil - i}}{\binom{n}{\lceil n(1/2+\eps)\rceil}}.$$
For each $i \in \{0,\ldots,\lceil t/2 \rceil - 1\}$, we have 
\begin{eqnarray*}
 \frac{\binom{n - t}{\lceil n(1/2+\eps)\rceil - i}}{\binom{n}{\lceil n(1/2+\eps)\rceil}}
&=& \frac{(n-t)!}{n!} \cdot \frac{\lceil n(1/2+\eps)\rceil!}{(\lceil n(1/2+\eps)\rceil - i)!}
\cdot \frac{\lfloor n(1/2-\eps)\rfloor!}{(\lfloor n(1/2-\eps)\rfloor - t + i)!}\\
&\ge & n^{-t} \cdot (n(1/2 + \eps - i/n))^i \cdot (n(1/2 - \eps - (t-i+1)/n))^{t-i}\\
&=&  2^{-t} \cdot (1 + 2 \eps - i/n)^i \cdot (1 - 2 \eps - (t-i+1)/n)^{t-i}\\
&\ge& 2^{-t}\cdot (1 + 2 \eps - (t+1)/n)^i \cdot (1 - 2 \eps - (t+1)/n)^{t-i}.
\end{eqnarray*}
Since $\eps < 1/8$, $1-2\eps > 3/4$. Since $t \le \sqrt{n/3}-1$, we have $(t+1)^2 \le n/3$ and therefore,
$(t+1)/n \le 1/(3(t+1)) \le 1/(3t)$. We have $((t+1)/n)/(1-2\eps) \le 4/(9t)$.  This implies both
$$(1-2\eps - (t+1)/n) \ge (1-2\eps) \cdot (1-4/(9t)),$$
and  
$$(1+2\eps - (t+1)/n) \ge (1+2\eps) \cdot (1-4/(9t)).$$
We obtain
$$\frac{\binom{n - t}{\lceil n(1/2+\eps)\rceil - i}}{\binom{n}{\lceil n(1/2+\eps)\rceil}}
\ge 2^{-t}\cdot (1 + 2 \eps)^i \cdot (1 - 2 \eps)^{t-i} \cdot (1-4/(9t))^t.
$$
One can show that for $\delta \in [0,1/2]$, $1-\delta \ge e^{-2\delta}$. Hence
$$(1-4/(9t))^t \ge e^{-2\cdot\frac{4}{9t} \cdot t} \ge 1/4,$$
and therefore,
$$\frac{\binom{n - t}{\lceil n(1/2+\eps)\rceil - i}}{\binom{n}{\lceil n(1/2+\eps)\rceil}}
\ge 2^{-t-2}\cdot (1 + 2 \eps)^i \cdot (1 - 2 \eps)^{t-i}.$$
We plug this bound into the inequality from \cite{CEG95} and obtain
\begin{eqnarray*}
\delta &\ge& 2^{-t-2} \sum_{i=0}^{\lceil t/2 \rceil - 1} \binom{t}{i} \cdot (1+2\eps)^i \cdot (1-2\eps)^{t-i}\\
&\ge&  2^{-t-2}
\cdot (1+2\eps)^{\lceil t/2 \rceil - \lceil \sqrt{t/2} \rceil} \cdot (1-2\eps)^{\lceil t/2 \rceil + \lceil \sqrt{t/2} \rceil} \cdot
\sum_{i=\lceil t/2 \rceil - \lceil \sqrt{t/2} \rceil}^{\lceil t/2 \rceil - 1} \binom{t}{i}\\
&\ge&  2^{-t-2}
\cdot (1-4\eps^2)^{\lceil t/2 \rceil - \lceil \sqrt{t/2} \rceil} \cdot (1-2\eps)^{2\lceil \sqrt{t/2} \rceil}
\cdot {\lceil \sqrt{t/2} \rceil} \cdot \binom{t}{\lceil t/2 \rceil - \lceil \sqrt{t/2} \rceil}\\
&\ge& 
4 \cdot e^{-8\eps^2(\lceil t/2 \rceil - \lceil \sqrt{t/2} \rceil)} \cdot
e^{-8\eps\lceil \sqrt{t/2} \rceil}
\cdot \frac{\lceil \sqrt{t/2} \rceil}{2^t} \cdot \binom{t}{\lceil t/2 \rceil - \lceil \sqrt{t/2} \rceil}.\\
\end{eqnarray*}
Using Stirling's approximation $\sqrt{2\pi} k^{k+1/2} e^{-k+1/(12k+1)}  < k! < \sqrt{2\pi} k^{k+1/2} e^{-k+1/(12k)}$, one can show that there is a positive constant $C_1$
such that 
$$\binom{t}{\lceil t/2 \rceil - \lceil \sqrt{t/2} \rceil} \ge C_1 \cdot 2^t / \sqrt{t}.$$
Plugging this into the previous inequality, we obtain for some positive constant $C_2$,
$$\delta \ge C_2 \cdot \exp\left(-8\eps^2(\lceil t/2 \rceil - \lceil \sqrt{t/2} \rceil) 
-8\eps\lceil \sqrt{t/2} \rceil\right).$$
This shows that for very small $\delta$ (namely, for $\delta < C_2/C_3$, where $C_3$ is a sufficiently large constant), $t > C_4 \cdot \frac{1}{\eps^2} \cdot \log (1/\delta)$, where $C_4$ is a positive constant. 

Note that even if $\delta$ is a relatively large constant less than $1/2$ ($5/12$ in our case),
$t > C_5 \cdot \frac{1}{\eps^2}$, for some positive $C_5$. This is the case, because if we had a better dependence on $\eps$ in this case, we could obtain a better dependence on $\eps$ also for small $\delta$ by routinely amplifying the probability of success of the algorithm, which incurs an additional multiplicative factor of only $O(\log(1/\delta))$. This finishes the proof.
\end{proof}

The above lemma shows a lower bound on the maximum number of queries. In the following corollary we extend the bound to the expected number of queries.

\begin{corollary}[Based on Corollary~2 of \cite{CEG95}]\label{cor:lb_avg_query}
Let $\eps \in (0,1/8)$, and let $n > 1/\eps^4$. Let $f$ be a function from $[n]$ to $\zo$. 
Let $\mathcal A$ be a randomized algorithm 
that  outputs an estimate $\bar\sigma$ to $\sigma \eqdef \frac{1}{n} \cdot \sum_{x \in [n]} f(x)$.

If $|\bar\sigma - \sigma| < \eps$ with probability at least $2/3$, then 
the expected number of queries of $\mathcal A$
to $f$ is at least $C/\eps^2$ for some function $f$, where $C$ is an absolute positive constant.
\end{corollary}

\begin{proof}
Let $t$ be the maximum expected number of queries of $\mathcal A$ to $f$, where the maximum is taken over all functions $f:[n]\to\zo$. Consider an algorithm $\mathcal A'$ that does the following. It simulates $\mathcal A$ until $\mathcal A$ attempts to make a $(\lfloor 12t \rfloor + 1)$-th query.
In this case $\mathcal A'$ interrupts the execution of $\mathcal A$, and outputs 0.
Otherwise $\mathcal A'$ returns the output of $\mathcal A$.
The probability that $\mathcal A'$ returns an incorrect answer is bounded by $1/3 + 1/12 = 5/12$.
By Lemma~\ref{lem:lb_worst_query}, $\mathcal A'$ makes at least
$C_1 / \eps^2$ queries, where $C_1$ is a positive constant.
Hence
$12t > C_1 / \eps^2$, which proves the claim.
\end{proof}

Finally we show a bound on the expectation of $\sum_i \mathbb E[w_i]$. The bound uses the fact that $w_i$'s have to be distributed in such a way that we are able to both observe many small $a_i$'s and few large $a_i$'s. Intuitively, there are roughly $\Theta(\log n)$ different possible magnitudes of $a_i$'s, and $w_i$'s of different size must be used to efficiently observe a sufficiently large number of $a_i$'s of each magnitude. This yields an additional logarithmic factor in the lower bound.

\begin{proof}[Proof of Theorem~\ref{thm:lowerBound}]
Consider the case of $\sigma$ between $0$ and $\rho/\eps$. If $\hat \sigma$ is a $(\rho,e^\eps)$ estimator for $\sigma$, then 
$$\sigma \cdot e^{-\eps} - \rho \le \hat\sigma \le \sigma \cdot e^{\eps} + \rho,$$
$$\sigma \cdot (1-\eps) - \rho < \hat\sigma < \sigma \cdot (1+2\eps) + \rho,$$
$$\sigma- 2\rho < \hat\sigma < \sigma + 3\rho,$$
$$|\sigma - \hat\sigma| < 3\rho.$$
Therefore, the estimator is also an additive approximation for $\sigma$.

Consider an integer $j$ such that $(\rho/\eps) \le 2^j$ and $(6\eps)^{-4} <2^j \le n$. We create a sequence $\{a_i\}$ as follows. Let $f$ be a function from $[2^j] \to \zo$. We select a subset $\mathcal I \subseteq [n]$ of size $2^j$ uniformly at random. For each $i \not\in \mathcal I$, we set $a_i = 0$. For $i \in \mathcal I$, we set $a_i = (1+f(k))/2 \cdot (\rho/\eps)/2^j$, where $k$ is the rank of $i$ in $\mathcal I$. We have $\sigma = \frac{\rho}{2\eps}(1+2^{-j}\sum_{x \in [2^j]} f(x))$.
Therefore, $R$ has to compute an additive $3\rho / (\rho/(2\eps)) = 6\eps$ approximation to $2^{-j}\sum_{x \in [2^j]} f(x)$ with probability at least $2/3$, where the probability is taken over the random bits of $R$ and the random choice of $\{w_i\}$.

We now create a corresponding sequence $\{\hat a_i\}$. For $i \not \in \mathcal I$, we set $\hat a_i = 0$. For $i \in \mathcal I$, if $1/w_i < \frac{\rho}{2^{j+1}\eps}$, we set $\hat a_i = a_i$, and $\hat a_i = \frac{3\rho}{4\eps}$, otherwise. Effectively, $R$ can only see the values $f(k)$ for $k$ such that $1/w_i < \frac{\rho}{2^{j+1}\eps}$, where $i$ is the item of rank $k$ in $\mathcal I$.
Let $E_j$ be the expected number of indexes $i$ for which $w_i > 2^{j+1} \frac{\eps}{\rho}$. The expected number of values of $f$ that $R$ can see is then $\frac{2^j}{n} \cdot E_j$.
By Corollary~\ref{cor:lb_avg_query},
$$\frac{2^j}{n} \cdot E_j \ge C_1/(6\eps)^2$$
where $C_1$ is an absolute positive constant.
Therefore, 
$$ E_j \ge \frac{C_2 n}{2^j \eps^2}$$
for another absolute constant $C_2$.

Consider now the expectation of the sum of all $w_i$'s:
\begin{eqnarray*}
\sum_{i\in[n]} \EE{}{w_i} &\ge& \sum_{j \in \mathbb Z}\frac{2^{j+1}\eps}{\rho} \cdot\mathbb E\left[\#i:w_i\in\left(\frac{2^{j+1}\eps}{\rho},\frac{2^{j+2}\eps}{\rho}\right]\right]\\
&\ge& \sum_{j \in \mathbb Z}\frac{2^{j}\eps}{\rho} \cdot\mathbb E\left[\#i:w_i>\frac{2^{j+1}\eps}{\rho}\right]\\
&\ge& \sum_{j: \max\{\rho/\eps,(6\eps)^{-4}\} < 2^j \le n}\frac{2^{j}\eps}{\rho} \cdot E_j\\
&\ge& \sum_{j: \max\{\rho/\eps,(6\eps)^{-4}\} < 2^j \le n}\frac{2^{j}\eps}{\rho} \cdot \frac{C_2 n}{2^j \eps^2}\\
&\ge& \frac{C_2 n}{\rho \eps} \cdot \left(\lfloor\log n\rfloor - \lfloor\max\{\rho/\eps,(6\eps)^{-4}\} + 1\rfloor\right) \ge \frac{C_3 n}{\rho \eps} \log (n/\alpha),
\end{eqnarray*}
where $C_3$ is a fixed positive constant. This finishes the proof.

\end{proof}

\section*{Acknowledgments}

We would like to thank Piotr Indyk, Assaf Naor, and David Woodruff for
helpful discussions about the $F_k$-moment estimation problem. We also
thank Andrew McGregor for kindly giving an overview of the landscape of the
heavy hitters problem.

\bibliographystyle{alphainit}
\bibliography{bibfile} 

\newcommand{\etalchar}[1]{$^{#1}$}
\def\cprime{$'$} \def\cprime{$'$}
\begin{thebibliography}{KNPW11}

\bibitem[ADIW09]{ADIW-EMD}
A.~Andoni, K.~{Do Ba}, P.~Indyk, and D.~Woodruff.
\newblock Efficient sketches for {Earth-Mover} {Distance}, with applications.
\newblock In {\em Proceedings of the Symposium on Foundations of Computer
  Science (FOCS)}, 2009.

\bibitem[AKO10]{AKO-edit}
A.~Andoni, R.~Krauthgamer, and K.~Onak.
\newblock Polylogarithmic approximation for edit distance and the asymmetric
  query complexity.
\newblock In {\em Proceedings of the Symposium on Foundations of Computer
  Science (FOCS)}, 2010.
\newblock A full version is available at
  \texttt{http://arxiv.org/abs/1005.4033}.

\bibitem[AMS99]{AMS}
N.~Alon, Y.~Matias, and M.~Szegedy.
\newblock The space complexity of approximating the frequency moments.
\newblock {\em J. Comp. Sys. Sci.}, 58:137--147, 1999.
\newblock Previously appeared in STOC'96.

\bibitem[BG06]{BG-entropy06}
L.~Bhuvanagiri and S.~Ganguly.
\newblock Estimating entropy over data streams.
\newblock In {\em Proceedings of the European Symposium on Algorithms (ESA)},
  pages 148--159, 2006.

\bibitem[BGKS06]{Ganguly-Fk06}
L.~Bhuvanagiri, S.~Ganguly, D.~Kesh, and C.~Saha.
\newblock Simpler algorithm for estimating frequency moments of data streams.
\newblock In {\em Proceedings of the ACM-SIAM Symposium on Discrete Algorithms
  (SODA)}, pages 708--713, 2006.

\bibitem[BJKS04]{BJKS}
Z.~{Bar-Yossef}, T.~S. Jayram, R.~Kumar, and D.~Sivakumar.
\newblock An information statistics approach to data stream and communication
  complexity.
\newblock {\em J. Comput. Syst. Sci.}, 68(4):702--732, 2004.

\bibitem[BO10a]{BO-independence10}
V.~Braverman and R.~Ostrovsky.
\newblock Measuring independence of datasets.
\newblock In {\em Proceedings of the Symposium on Theory of Computing (STOC)},
  2010.

\bibitem[BO10b]{BO10}
V.~Braverman and R.~Ostrovsky.
\newblock Recursive sketching for frequency moments.
\newblock {\em CoRR}, abs/1011.2571, 2010.

\bibitem[BO10c]{BO-zeroonelaw10}
V.~Braverman and R.~Ostrovsky.
\newblock Zero-one frequency laws.
\newblock In {\em Proceedings of the Symposium on Theory of Computing (STOC)},
  2010.

\bibitem[CCFC02]{CCF}
M.~Charikar, K.~Chen, and M.~Farach-Colton.
\newblock Finding frequent items in data streams.
\newblock In {\em Proceedings of International Colloquium on Automata,
  Languages and Programming (ICALP)}, 2002.

\bibitem[CDK{\etalchar{+}}09]{CDKLT09-streamSampling}
E.~Cohen, N.~G. Duffield, H.~Kaplan, C.~Lund, and M.~Thorup.
\newblock Stream sampling for variance-optimal estimation of subset sums.
\newblock In {\em Proceedings of the ACM-SIAM Symposium on Discrete Algorithms
  (SODA)}, pages 1255--1264, 2009.

\bibitem[CEG95]{CEG95}
R.~Canetti, G.~Even, and O.~Goldreich.
\newblock Lower bounds for sampling algorithms for estimating the average.
\newblock {\em Inf. Process. Lett.}, 53(1):17--25, 1995.

\bibitem[CKS03]{ChakrabartiKS03}
A.~Chakrabarti, S.~Khot, and X.~Sun.
\newblock Near-optimal lower bounds on the multi-party communication complexity
  of set disjointness.
\newblock In {\em IEEE Conference on Computational Complexity}, pages 107--117,
  2003.

\bibitem[CM05a]{CM05b}
G.~Cormode and S.~Muthukrishnan.
\newblock An improved data stream summary: the count-min sketch and its
  applications.
\newblock {\em J. Algorithms}, 55(1):58--75, 2005.
\newblock Previously in LATIN'04.

\bibitem[CM05b]{CM05}
G.~Cormode and S.~Muthukrishnan.
\newblock Space efficient mining of multigraph streams.
\newblock In {\em Proceedings of the ACM Symposium on Principles of Database
  Systems (PODS)}, 2005.

\bibitem[CM10]{countminsketch}
G.~Cormode and M.~Muthukrishnan.
\newblock Count-min sketch.
\newblock 2010.
\newblock \texttt{https://sites.google.com/site/countminsketch}.

\bibitem[DLT07]{DLT07-prioritySampling}
N.~G. Duffield, C.~Lund, and M.~Thorup.
\newblock Priority sampling for estimation of arbitrary subset sums.
\newblock {\em J. ACM}, 54(6), 2007.

\bibitem[Gan11]{Ganguly-personal}
S.~Ganguly.
\newblock Personal communication.
\newblock April 2011.

\bibitem[GBD08]{GBD-hybrid08}
S.~Ganguly, M.~Bansal, and S.~Dube.
\newblock Estimating hybrid frequency moments of data streams.
\newblock In {\em Frontiers in Algorithmics}, 2008.

\bibitem[GC07]{GC-moments07}
S.~Ganguly and G.~Cormode.
\newblock On estimating frequency moments of data streams.
\newblock In {\em Proceedings of the International Workshop on Randomization
  and Computation (RANDOM)}, pages 479--493, 2007.

\bibitem[Ind06]{I00b}
P.~Indyk.
\newblock Stable distributions, pseudorandom generators, embeddings and data
  stream computation.
\newblock {\em J. ACM}, 53(3):307--323, 2006.
\newblock Previously appeared in FOCS'00.

\bibitem[Iof10]{I10-consSampling}
S.~Ioffe.
\newblock Improved consistent sampling, weighted minhash and l1 sketching.
\newblock In {\em International Conference on Data Mining}, 2010.

\bibitem[IW05]{IW05}
P.~Indyk and D.~Woodruff.
\newblock Optimal approximations of the frequency moments of data streams.
\newblock {\em Proceedings of the Symposium on Theory of Computing (STOC)},
  2005.

\bibitem[JST10]{JST10}
H.~Jowhari, M.~Saglam, and G.~Tardos.
\newblock Tight bounds for lp samplers, finding duplicates in streams, and
  related problems.
\newblock {\em CoRR}, abs/1012.4889, 2010.

\bibitem[JW09]{JW-cascaded}
T.~Jayram and D.~Woodruff.
\newblock The data stream space complexity of cascaded norms.
\newblock In {\em Proceedings of the Symposium on Foundations of Computer
  Science (FOCS)}, 2009.

\bibitem[KNPW11]{KNPW10-otimalUpdate}
D.~M. Kane, J.~Nelson, E.~Porat, and D.~P. Woodruff.
\newblock Fast moment estimation in data streams in optimal space.
\newblock In {\em Proceedings of the Symposium on Theory of Computing (STOC)},
  2011.
\newblock A previous version appeared as ArXiv:1007.4191,
  \texttt{http://arxiv.org/abs/1007.4191}.

\bibitem[KNW10]{KNW-exact10}
D.~M. Kane, J.~Nelson, and D.~P. Woodruff.
\newblock On the exact space complexity of sketching small norms.
\newblock In {\em Proceedings of the ACM-SIAM Symposium on Discrete Algorithms
  (SODA)}, 2010.

\bibitem[Li08]{Li-estimators08}
P.~Li.
\newblock Estimators and tail bounds for dimension reduction in $l_p$ $(0 < p
  \le 2)$ using stable random projections.
\newblock In {\em Proceedings of the ACM-SIAM Symposium on Discrete Algorithms
  (SODA)}, 2008.

\bibitem[MG82]{MG82}
J.~Misra and D.~Gries.
\newblock Finding repeated elements.
\newblock {\em Sci. Comput. Program.}, 2(2):143--152, 1982.

\bibitem[Mut05]{Muthu-book}
M.~Muthukrishnan.
\newblock {\em Data Streams: Algorithms and Aplications}.
\newblock Foundations and Trends in Theoretical Computer Science. Now
  Publishers Inc, January 2005.

\bibitem[MW10]{MW-l1sampling}
M.~Monemizadeh and D.~Woodruff.
\newblock 1-pass relative-error $l_p$-sampling with applications.
\newblock In {\em Proceedings of the ACM-SIAM Symposium on Discrete Algorithms
  (SODA)}, 2010.

\bibitem[NW10]{NW-fastsketches10}
J.~Nelson and D.~Woodruff.
\newblock Fast manhattan sketches in data streams.
\newblock In {\em Proceedings of the ACM Symposium on Principles of Database
  Systems (PODS)}, 2010.

\bibitem[Sze06]{Sze06-DLPoptimality}
M.~Szegedy.
\newblock The dlt priority sampling is essentially optimal.
\newblock In {\em Proceedings of the Symposium on Theory of Computing (STOC)},
  pages 150--158, 2006.

\end{thebibliography}

\appendix

\section{Bound on $\EE{w\sim \W}{w^\alpha}$}
\label{apx:expAlpha}

\begin{claim}
For $k\ge 1$, suppose $u_j$ are drawn uniformly at random from $[0,1]$. Then, for any $\alpha\in(0,1)$, we have that
$$ \EE{u_j}{\left(\max_j 1/u_j\right)^\alpha} \le
O\left(\tfrac{k^{\alpha}}{1-\alpha}\right).
$$
\end{claim}
\begin{proof}
We compute the expectation directly:
\begin{eqnarray*}
\EE{u_j}{\left(\max_j 1/u_j\right)^\alpha}
& = &
\int_{0}^1 u^{-\alpha}\cdot k(1-u)^{k-1}\du
\\
&\le&
\int_{0}^{1/k} k\cdot u^{-\alpha}\du
+
\int_{1/k}^1 k^\alpha \cdot k(1-u)^{k-1}\du
\\
& = &
k\cdot \Big[ \tfrac{u^{1-\alpha}}{1-\alpha} \Big]_0^{1/k}
+
k^\alpha \Big[ -(1-u)^k \Big]_{1/k}^1 
\\
&\le&
O(\tfrac{k^\alpha}{1-\alpha}).
\end{eqnarray*}
\end{proof}

\end{document}